\documentclass[acmsmall,nonacm]{acmart}

\usepackage{tikz}
\usepackage{algpseudocode}
\usepackage[table]{xcolor}

\newcommand{\name}{Wally}

\newcommand{\swifthe}{Swift Homomorphic Encryption}
\newcommand{\swifthelink}{https://github.com/apple/swift-homomorphic-encryption}
\newcommand{\swifthecitation}{Apple_Inc_and_Swift_Homomorphic_Encryption_project_authors_Swift_Homomorphic_Encryption}

\newcommand{\Zset}{\mathbb{Z}}

\newcommand{\kg}{\mathsf{KeyGen}}
\newcommand{\enc}{\mathsf{Enc}}

\newcommand{\sk}{\mathsf{sk}}
\newcommand{\evk}{\mathsf{evk}}
\newcommand{\ct}{\mathsf{ct}}
\newcommand{\id}{\mathsf{id}}

\renewcommand{\vec}[1]{\mathbf{#1}}
\newcommand{\mat}[1]{\mathbf{#1}}

\begin{document}

\title{Wally: Batched Private Nearest Neighbor Search at Scale}

\author{Hilal Asi}
\affiliation{\institution{Apple Inc.}\country{USA}}
\email{hasi@apple.com}

\author{Fabian Boemer}
\affiliation{\institution{Apple Inc.}\country{USA}}
\email{fboemer@apple.com}

\author{Nicholas Genise}
\affiliation{\institution{Apple Inc.}\country{USA}}
\email{ngenise@apple.com}

\author{Muhammad Haris Mughees}
\authornote{Corresponding author.}
\affiliation{\institution{Apple Inc.}\country{USA}}
\email{mughees@apple.com}

\author{Tabitha Ogilvie}
\authornote{This work was done while interning at Apple.}
\affiliation{\institution{Royal Holloway, University of London}\country{United Kingdom}}
\email{tabitha.l.ogilvie@gmail.com}

\author{Rehan Rishi}
\affiliation{\institution{Apple Inc.}\country{USA}}
\email{rrishi@apple.com}

\author{Guy N. Rothblum}
\affiliation{\institution{Apple Inc.}\country{USA}}
\email{gn_rothblum@apple.com}

\author{Kunal Talwar}
\affiliation{\institution{Apple Inc.}\country{USA}}
\email{ktalwar@apple.com}

\author{Karl Tarbe}
\affiliation{\institution{Apple Inc.}\country{USA}}
\email{ktarbe@apple.com}

\author{Akshay Wadia}
\affiliation{\institution{Apple Inc.}\country{USA}}
\email{awadia@apple.com}

\author{Ruiyu Zhu}
\affiliation{\institution{Apple Inc.}\country{USA}}
\email{rzhu@apple.com}

\author{Marco Zuliani}
\affiliation{\institution{Apple Inc.}\country{USA}}
\email{mzuliani@apple.com}

\renewcommand{\shortauthors}{Asi et al.}



\begin{abstract}
We present Wally, a batched private nearest-neighbor search protocol that uses differential privacy to break the linear computation barrier of fully-oblivious schemes. In prior systems like Tiptoe, the server must process the entire database per query to hide the access pattern, resulting in low throughput (909 QPS) and high communication (17.4~MB) on a 3.2-million-entry database. Sublinear alternatives like Pacmann avoid full scans but require 614~MB of client storage and an offline phase where clients stream the entire database.

Wally's key insight is that fully-oblivious schemes are prohibitively expensive at scale, but the same scale that makes them expensive also provides an opportunity. Large-scale systems naturally have many concurrent clients. Wally batches queries from these non-coordinating clients and has each client independently add a few fake queries to hide which clusters it accesses. The fake query counts follow a negative binomial distribution, which is both non-negative (valid as a query count) and infinitely divisible (allowing each client to sample independently without coordination). Each client sends its queries at random times through an existing anonymization service, avoiding the need for a centralized shuffler and eliminating any single point of failure. The server sees only an anonymized, noisy stream of cluster accesses that is provably $(\epsilon, \delta)$-differentially private, enabling it to compute over only the relevant clusters rather than the entire database. To further reduce communication, the client encrypts its query embedding under somewhat homomorphic encryption (SHE) so that the server returns only encrypted similarity scores rather than full cluster contents. On a 3.2-million-entry database with batches of 500,000 queries, Wally achieves $7\text{-}29\times$ higher throughput and $6.7\text{-}31\times$ lower communication than Tiptoe, and $15{,}000\times$ lower client storage than Pacmann, while providing strong $(\epsilon{=}0.1, \delta{=}2^{-26})$-differential privacy and comparable or better accuracy.

For metadata retrieval, the client uses keyword private information retrieval (PIR) to privately fetch data associated with the nearest result. We propose optimizations to both SHE and keyword PIR that yield $2\text{-}3\times$ improvements in PIR and $20\text{-}25$\% in BFV operations, and release a production-quality open-source BFV homomorphic encryption library in Swift.
\end{abstract}

\maketitle

\section{Introduction}\label{sec:intro}
Nearest neighbor search is a key building block in many applications. Search engines use it to find relevant documents~\cite{carrara2022approximate,zhang2019grip}. Recommendation systems use it to suggest products~\cite{ahuja2019movie}. Retrieval-Augmented Generation (RAG) systems use it to retrieve relevant context that augments language models' responses to user prompts~\cite{sawarkar2024blended, csakar2025maximizing}. In this scenario, a server maintains a database of semantically significant vectors, while a client sends a query vector. The server then identifies the vectors in its database that are most similar to the client's query. However, this setup raises privacy concerns, as the server learns the client's query, potentially exposing sensitive personal information. For example, if the client query is an embedding extracted from a personal photo, the server can potentially invert this embedding and infer sensitive information about the client~\cite{DBLP:conf/emnlp/MorrisKSR23}.

The potential privacy issues surrounding nearest-neighbor search have motivated numerous private nearest-neighbor (PNN) protocols. Legacy PNN protocols provide privacy by periodically injecting fake queries that resemble real queries~\cite{xu2021utilitarian,petit2015peas, nissenbaum2009trackmenot, murugesan2009providing, yang2017anserini} or by obfuscating clients' original queries using fake terms~\cite{DBLP:journals/pvldb/PangDX10, ye2009noise, rebollo2010optimized, arampatzis2015versatile}. Although these approaches are efficient, they lack standard privacy guarantees and often rely on assumptions about limited adversarial capabilities. For example, many of these works assume that adversaries have limited knowledge, such as being unaware of the obfuscation algorithm, lacking background data, or employing fixed inference strategies. Consequently, when removing these assumptions or considering a more capable adversary, their privacy guarantee breaks~\cite{al2012trackmenot,balsa2012ob, gervais2014quantifying, peddinti2010privacy}.

Although recent PNN protocols offer strong privacy guarantees, they rely on restrictive assumptions or incur prohibitive costs. For example, PRECO~\cite{DBLP:conf/sp/Servan-Schreiber22} and SANNS~\cite{DBLP:conf/uss/0030CDPRR20} achieve efficiency by assuming the existence of two non-colluding servers. This requires two independent organizations to jointly run expensive computation, and if the servers collude, all privacy guarantees are lost with no graceful degradation. In practice, finding and maintaining two independent, trusted, and computationally capable parties is a significant deployment barrier.

In contrast, Tiptoe~\cite{DBLP:conf/sosp/HenzingerDCZ23}, a state-of-the-art single-server protocol which uses somewhat homomorphic encryption (SHE), requires approximately 17 MB of communication per query and achieves only 909 QPS on 10,000 cores for a database with 3.2 million entries. This overhead arises because the server needs to scan the entire dataset, and the client must send cryptographic material for each entry. More recent sublinear computation approaches, such as Pacmann~\cite{DBLP:journals/iacr/ZhouSF24}, eliminate full database scans at query time but require clients to stream the entire database as a pre-computation, resulting in 61 MB of amortized communication per query for the same database. Such overheads are prohibitive in production deployments, which are critical for recent AI applications where systems must serve thousands of concurrent clients over a vector index spanning millions of vectors.

As a result, current PNN protocols face an undesirable trade-off between privacy and performance. On one hand, efficient approaches lack a standard privacy guarantee, rendering them vulnerable. On the other hand, schemes that do offer a standard privacy guarantee impose extremely high computational or communication requirements, making them impractical for systems that have high query volume.

\textbf{Our System.}
To bridge this gap, we present \name, a batched private nearest-neighbor protocol that turns high query volume into a performance advantage. The key insight is that differential privacy enables the server to compute over only the queried clusters rather than the entire database, breaking the linear computation barrier of fully-oblivious schemes like Tiptoe. Unlike prior sublinear schemes such as Pacmann, which require clients to stream the entire database offline (614~MB of client storage) and perform many sequential round trips, \name\ requires no offline phase, no large client-side storage, and no interactive rounds.

\name\ batches queries from non-coordinating clients and processes them in fixed-duration epochs. Each client adds fake queries drawn from a negative binomial distribution. This distribution is a natural fit because fake query counts must be non-negative integers, ruling out continuous distributions like the Gaussian. Moreover, the negative binomial is infinitely divisible, which allows non-coordinating clients to independently generate noise whose aggregate provably provides $(\epsilon, \delta)$-DP, without any client-to-client communication. All queries (real and fake) are routed through an existing anonymization service that strips identifying information. Unlike traditional shuffle-DP protocols that require a dedicated shuffler to collect and permute all messages before forwarding, each client in \name\ independently sends queries at uniformly random time slots within the epoch. The anonymization network simply forwards queries as they arrive, and the shuffling emerges implicitly from the random timing and anonymization. This is significantly more practical than a centralized shuffler, which introduces a single point of trust and failure. If the shuffler goes down, the entire service halts. In \name, individual clients can join or leave without affecting other participants, and the protocol leverages widely deployed anonymization services rather than requiring purpose-built infrastructure. As a result, the server observes only a noisy histogram of cluster accesses that is statistically independent of any single client's queries. \name\ further utilizes SHE to compute similarity scores between encrypted query embeddings and cluster embeddings, reducing the response to only the scores rather than the full cluster contents. Once a client identifies its nearest neighbor, it privately fetches the associated metadata using keyword private information retrieval (keyword PIR). We propose novel optimizations to both SHE and keyword PIR that yield $2\text{-}3\times$ improvements.

\name\ achieves $7\text{-}29\times$ higher QPS and $6.7\text{-}31\times$ lower communication than Tiptoe, and $15{,}000\times$ lower client storage than Pacmann. The protocol provides a tunable accuracy-performance tradeoff at fixed differential privacy, allowing system designers to improve accuracy with graceful performance degradation.

\name\ trades full obliviousness for $(\epsilon,\delta)$-differential privacy, a weaker but widely accepted privacy standard. While fully-oblivious schemes like Tiptoe and Pacmann reveal nothing about the queried cluster, this guarantee comes at the cost of linear computation or massive client storage. In contrast, \name's DP guarantee bounds the server's ability to infer any single client's queries to a multiplicative factor of $e^{\epsilon}$, except with probability $\delta$. We set $\epsilon = 0.1$ and $\delta = 2^{-26}$, which represent a strong privacy guarantee. This relaxation is precisely what enables sublinear server cost, and we argue it is a favorable tradeoff for high-throughput applications where the performance cost of full obliviousness is prohibitive.

\textbf{Motivational example.} Our PNN is suitable for a wide range of applications. Here, we provide a concrete example from AI-based photo systems. Clients want to privately annotate their photos with relevant context retrieved from a knowledge database stored on a server. By utilizing the enhanced photo context, on-device models can generate better search results and recommendations. At scale, these systems have millions of users.
These retrievals can occur in the background since the enhanced search experience doesn't need to be immediately available.
In other words, high throughput is a must, while low latency can be relaxed.
Therefore, batch processing is natural.
Clients submit their photo embeddings to a batch, which is processed during scheduled epochs. Once the relevant context is retrieved and integrated, the enhanced systems provide more accurate photo search and personalized recommendations.

The protocol's privacy guarantee necessitates that the anonymization service and the server remain independent. We contend that our approach is more practical than a generic 2PC-based solution~\cite{DBLP:conf/sp/Servan-Schreiber22}. This is because numerous independent applications rely on anonymization services, and their security models depend on the trustworthiness of these services. Consequently, the anonymization network has a significant economic incentive to avoid colluding with any single partner, including our server.

\textbf{Our contributions.} We make the following contributions:
\begin{enumerate}
    \item We propose \name, a batched private nearest-neighbor protocol with sublinear server computation, no offline phase, no large client-side storage, and no interactive rounds. Prior sublinear schemes such as Pacmann require clients to stream the entire database offline and perform many sequential round trips. \name\ instead uses differential privacy to hide which clusters were accessed, enabling the server to compute over only the queried clusters. To our knowledge, this is the first system where differential privacy directly enables efficient private search.

    \item We introduce a distributed noise mechanism based on infinitely divisible distributions (negative binomial) that allows non-coordinating clients to independently generate fake queries whose aggregate provably provides $(\epsilon, \delta)$-DP. This requires no client-to-client communication or synchronization beyond a shared epoch structure. We provide a formal security proof demonstrating the privacy of the protocol against semi-honest adversaries.

    \item We propose several optimizations to keyword PIR and SHE computation that are of independent interest. For keyword PIR, we introduce two-table cuckoo hashing, optimized database layouts using uneven dimensions, and evaluation key reduction, yielding a $2\text{-}3\times$ improvement in server computation and communication. For SHE, we develop lazy modular reduction in the NTT and ciphertext operations, improving BFV runtime by $20\text{-}25$\%, and a plaintext RNS technique that enables high-precision similarity computation without increasing request size.

    \item We implemented the BFV SHE scheme in Swift and released it as a standalone open-source library. This is, to our knowledge, the first production-quality homomorphic encryption library in Swift, and it incorporates all of the above optimizations. The library is highly versatile and can be integrated into other privacy-preserving applications.

    \item We evaluate \name\ against state-of-the-art private search systems and demonstrate that \name\ achieves $7\text{-}29\times$ higher QPS and $6.7\text{-}31\times$ lower communication than Tiptoe, and $15{,}000\times$ lower client storage than Pacmann, while maintaining strong $(\epsilon{=}0.1, \delta{=}2^{-26})$-DP and comparable or better accuracy.
\end{enumerate}

\subsection{Overview}

At a high level, \name\ uses a clustering-based approximate nearest-neighbor scheme. This is the same search procedure implemented in modern vector search libraries and RAG retrieval stacks (e.g., FAISS IVF-Flat/IVF-PQ, Milvus IVF, Google ScaNN).

\textbf{Index construction (server side).}
The server hosts \(N\) assets and encodes each asset \(x\) into an embedding \(f(x)\in\mathbb{R}^d\) using an embedding model \(f\). The embedding function is semantic: the inner product between vectors of similar assets is large (equivalently, these vectors have high cosine similarity after normalization). The server then partitions the \(N\) embeddings into \(K\) clusters (e.g., via \(k\)-means clustering) and shares the \(K\) centroids \(\{c_1,\ldots,c_K\}\) of these clusters with the participating clients.

\textbf{Nearest-neighbor querying.}
A non-private search operates as follows.
Given a query \(q\),
\begin{enumerate}
  \item The client computes its embedding \(f(q)\) and selects the \(\Delta\) nearest centroids to \(f(q)\).
  \item The client sends those \(\Delta\) cluster identifiers to the server.
  \item The server returns the candidate set consisting of the embeddings (and their IDs) in the selected clusters.
  \item The client re-ranks the candidates locally using exact similarity to obtain the true nearest neighbor.
\end{enumerate}

This protocol lacks privacy; by monitoring the cluster IDs of clients' query, the server gains insight into their inquiries. Since single-query PNN systems are inefficient, we consider PNN at the batch level.

\textbf{Wally: batched private nearest neighbor search.}
Let \(U\) represent the number of clients participating in a batch, each contributing a maximum of \(\Delta\) cluster queries. Our protocol ensures the server's batch-level observation is a noisy histogram of cluster counts with per-client sensitivity of \(\Delta\), yielding \((\varepsilon, \delta)\)-DP. The next paragraphs detail the construction.

In the first step of the construction, all queries are routed through an anonymization service. This strips queries of all identifiable information, such as IP addresses, before forwarding them to the \name\ server. We assume that the anonymity network combines queries that arrive within one-second slots and forwards them together.

Although the queries are anonymized, this does not guarantee privacy. The server can infer information about the client's query by analyzing the access pattern (e.g., whether a particular cluster is accessed). Therefore, in addition, each participating client also issues fake queries to obscure access patterns. The number of fake queries per client is distributed according to a negative binomial distribution. The expected number of fake queries is very small.
To generate a single fake query, the client selects a uniform random cluster and queries it. Importantly, the client does not need to coordinate this query with the server or with other clients. As shown in \autoref{sec:performance}, the overhead of fake queries is modest (on average one to two fake queries per real query). We emphasize that fake queries do not affect the accuracy or correctness seen by any client, as the client can discard these responses.

The protocol operates in fixed-duration \emph{epochs}, each of which defines a batch of queries from $U$ clients. Within an epoch, each client schedules queries at uniformly random time slots, preventing the server from linking queries to clients via timing. The per-client overhead of fake queries scales inversely with $U$, so high query volume directly reduces overhead.

We show that anonymizing queries, sampling fake queries from a negative binomial distribution, and randomizing the query schedule within an epoch are equivalent to sharing with the server a \emph{noised histogram} of queries from all clients. We then prove that this noisy histogram guarantees $(\epsilon, \delta)$-DP.

For fixed $U$ and $\Delta$, the expected number of fake queries increases linearly with the number of clusters $K$. Therefore, to reduce the overhead of fake queries, we set $K$ to a lower value. But doing so will increase cluster sizes, which in turn would lead to high response overhead, as the server returns embeddings for the entire cluster for each query. To reduce response overhead, we utilize somewhat homomorphic encryption (SHE). Specifically, for each real query, the client encrypts the query embedding using SHE and sends it, along with the corresponding cluster identifier, to the server. For a fake query, the client encrypts an arbitrary embedding and picks a cluster identifier uniformly at random.

The server computes the similarity by performing a dot product between the encrypted query embeddings and cluster embeddings under SHE. Conceptually, this requires multiplying an encrypted client vector with the server plaintext embedding matrix; therefore, we adopt the diagonal matrix encoding of~\cite{DBLP:conf/crypto/HaleviS18}.
For similarity computations, this encoding eliminates unnecessary (typically expensive) SHE operations, such as rotations and ciphertext-ciphertext multiplications. Furthermore, this requires only a single level of ciphertext-plaintext multiplication to compute similarity scores. We extend this encoding with efficient plaintext packing, enabling the packing of thousands of scores into a single SHE ciphertext.

\textbf{Private metadata fetch.}
In many applications, once the client has identified a nearest neighbor, it wants to fetch metadata for that asset. As mentioned, we rely on keyword private information retrieval (keyword PIR) for metadata fetch. At a high level, once all clients have identified the entries nearest to their embeddings, they participate in a batch keyword PIR protocol to retrieve metadata. Our batch keyword PIR follows the same blueprint as our batch PNN search protocol. The server divides the metadata into clusters. The client queries real and fake clusters at random instances. For each real query, the client generates a keyword PIR query for the identifier of the nearest neighbors as input. Similarly, for each fake query, the client can use any fake keyword to generate a keyword PIR query. The server then evaluates the keyword PIR queries against the desired clusters and returns encrypted responses containing the desired metadata to the client.

Our keyword PIR scheme uses cuckoo hashing to map keywords to buckets and then uses SHE-based index PIR scheme MulPIR to fetch particular buckets~\cite{DBLP:conf/uss/AliLP0SSY21}. We propose several novel optimizations that improve the computation and reduce the communication overhead by $2\text{-}3\times$. We believe that these optimizations are of independent interest.

\textbf{Implementation.} We implemented similarity computation and keyword PIR using the BFV SHE scheme in Swift. To optimize the implementation, we employed various techniques. For instance, in BFV SHE operations, integer modulus operations are bottlenecks. To mitigate this, we evaluate operations in large integer buffers and perform the modulus operation only when the buffer overflows. Additionally, when we evaluate the same operation multiple times in a sequence, we utilize hoisting (moving invariant computations outside loops) to avoid redundant steps. We have released all these optimizations along with our implementation of BFV as a standalone Swift library.

We discuss related works in \autoref{sec:relatedwork}.

\section{Background}\label{sec:background}

Let $[n] := \{1,\dots,n\}$.
Matrices are in bold capitals (e.g., $\mat M$) and vectors in bold lowercase (e.g., $\vec v$). All logarithms are natural unless stated. Sampling $a$ from a distribution $D$ is written $a \gets D$; for finite set $S$, $a \gets S$ denotes uniform sampling. We write $\Zset$ for integers and $\Zset_Q$ for integers modulo $Q$. For parameters $r,p$, $\mathcal{NB}(r,p)$ is the negative binomial distribution with pmf $\Pr[X=k]=\binom{k+r-1}{k}(1-p)^k p^r$, $k\ge0$, which is infinitely divisible (Def.~\ref{def:inf_div}).

\begin{definition}\label{def:inf_div}
    A distribution $\mathcal{D}$ is infinitely divisible if for all
    $n \in \mathbb{N}$, it can be expressed as the sum
    of $n$ i.i.d. variables. I.e., there exists a $\mathcal{D}'$ such
    that $\mathcal{X}_1 + \cdots + \mathcal{X}_n \sim \mathcal{D}$
    and $\mathcal{X}_i \sim \mathcal{D}'$.
\end{definition}

\textbf{Differential privacy}
Intuitively, a protocol between many clients and a server is differentially
private if it is hard to tell whether or not a
specific client is present in the protocol.
We use the standard notion of differential privacy
(DP)~\cite{DBLP:conf/tcc/DworkMNS06,DBLP:conf/eurocrypt/DworkKMMN06}: a randomized mechanism
$M$ is $(\epsilon, \delta)$-differentially private for $\epsilon \geq 0$ and
$\delta \in [0,1]$, if for any two datasets $X, X'$ differing on
one element, and for any subset $S$ of possible transcripts output by $M$,
$\Pr[M(X) \in S] \leq e^\epsilon \cdot \Pr[M(X') \in S] + \delta$
where the probability is over the randomness of $M$.
We use basic composition in DP~\cite[Lem. 2.3]{vadhan2017complexity}.
\begin{lemma}\label{lemma:basic_comp}
    If $M_1(X), \dots, M_k(X)$ are each $(\epsilon, \delta)$-DP
    with independent randomness, then $(M_i(X))_{i = 1}^k$ is $(k\epsilon,k\delta)$-DP.
\end{lemma}

\textbf{Anonymous networks}
Anonymous networks (ANs), also known as anonymization networks or anonymous communication protocols, provide user anonymity through cryptographic techniques such as onion routing~\cite{DBLP:journals/cacm/GoldschlagRS99} and mixnets~\cite{DBLP:journals/cacm/Chaum81}. These protocols route messages through multiple intermediaries, called \emph{mixes} or \emph{hops}. Throughout this paper, we assume that ANs (1) remove all identifying information from client queries and (2) batch queries within one-second intervals before forwarding them to the next hop.

\subsection{Somewhat Homomorphic Encryption (SHE)}
We focus on SHE schemes based on the ring learning
with errors (RLWE) problem. Concretely, \name\ can be implemented using
any SHE scheme such as BFV~\cite{DBLP:journals/iacr/FanV12} or BGV~\cite{DBLP:journals/toct/BrakerskiGV14}.
A plaintext $m$ is a polynomial in a ring $R_t = \Zset_t[X]/(X^n+1)$ (with degree at most $n-1$).
The secret key $s$ is a
polynomial of degree $n-1$ with small coefficients, in $\{0,\pm 1\}$.
A ciphertext is a pair of polynomials $\ct = (c_0,c_1) = (a, as+\lfloor Q/t\rfloor m+e) \in R_Q^2$
 where $Q$ is ciphertext modulus and $R_Q := \Zset_Q[X]/(X^n+1)$.
Here $a$ is picked uniformly at random and $e$ is a noise polynomial with coefficients
sampled from a centered binomial distribution. A scheme satisfies its decryption formula: $c_0 + c_1s \mod Q=  \lfloor Q/t\rfloor m + e$.
 This noise polynomial's coefficients grow as we compute more
homomorphic operations on the ciphertext, and we can only decrypt correctly
 if $\|e\|_\infty < Q/2t$.

\textbf{Vectorized SHE.}
When a plaintext modulus $t$ is a prime satisfying $t = 1 \mod 2n$,
BGV and BFV support operations on plaintext vectors of degree $n$ encoded as polynomials.
Here, each ciphertext encrypts a plaintext vector $\vec v \in \Zset_t^n$,
and the homomorphic operations are vector-wise.
Vectorized SHE is crucial for homomorphic linear algebra.

\textbf{SHE operations and costs.}
The SHE schemes we consider support the following operations (in vector and polynomial arithmetric). Let $\enc(\vec v)$ denote a valid encryption of $\vec v\in\mathbb Z^n_t$. Given $\ct = \enc(\vec v)$, $\ct' = \enc(\vec v')$,
\begin{itemize}
    \item $\mathsf{CtCtAdd}(\ct, \ct')$
        returns $\ct_+ = \enc(\vec v+\vec v')$.
    \item $\mathsf{PtCtMult}(\ct, \vec v')$
        returns $\ct_\times = \enc(\vec v\odot\vec v')$.
    \item $\mathsf{CtCtMult}(\ct, \ct', ek)$
        returns $\ct_\times = \enc(\vec v\odot\vec v')$.
    \item $\mathsf{CtRotate}(\ct, r, ek)$ for $r \in [0,n/2)$ returns a ciphertext encrypting
            $(\mathsf{Rot}(\vec v), \mathsf{Rot}(\vec{\bar v})) \in \Zset_t^{n/2} \times \Zset_t^{n/2}$
        where $\mathsf{Rot}(\vec v)$ ($\mathsf{Rot}(\vec{\bar v})$) is
        $\vec v$ ($\vec{\bar v}$) cyclically shifted to the left by
        $r$ positions, if $\ct$ originally encrypted $\vec v' = (\vec v, \vec{\bar v})$.
        We can also homomorphically swap
        $(\vec v, \vec{\bar v})$.
\end{itemize}

$\mathsf{CtCtAdd}$ and $\mathsf{PtCtMult}$ are efficient, with the latter incurring noise growth proportional to $t$.
$\mathsf{CtCtMult}$ is costlier, causing large noise growth from polynomial multiplication and reduction modulo $Q$.
$\mathsf{CtRotate}$ add little noise but are slow due to key-switching.
Thus, minimizing $\mathsf{CtCtMult}$ and $\mathsf{CtRotate}$ is key for scalability.

\textbf{Evaluation keys.}
$\mathsf{CtCtMult}$ and $\mathsf{CtRotate}$ require evaluation keys that depend on the client's secret key. Typically, the server could store the evaluation keys once in typical SHE applications and reuse them across various client requests. However, the server learns the requests belonging to the same client. Looking ahead, this will break the privacy guarantee in \name; therefore, the client must send fresh evaluation keys with each request.
These keys contribute to the client's request size.


\subsection{Threat Model}\label{ssec:threat_model}
\name~has three kinds of entities: clients, an anonymization service, and the server. Client queries must remain private from all other parties under the following threat assumptions.

 The server is semi-honest: it follows the protocol but attempts to infer queries from observed traffic. Malicious clients may deviate by withholding fake queries (described in \autoref{sec:overview}), which can weaken privacy guarantees for others.

We assume the server and anonymization service do not collude. This is practical since widely deployed anonymization systems (e.g., Tor, mix networks) already support privacy-critical applications and have strong disincentives to collude with individual servers.

Under these assumptions, \name~ensures differential privacy for honest clients’ queries: the server’s view is statistically indistinguishable regardless of whether a particular honest client participates.

\textbf{Mitigating malicious clients.} To address clients who withhold fake queries, honest clients can over-provision by generating more fake queries than the minimum needed. This increases each honest client's fake queries by a factor of $1/(1-\alpha)$, where $\alpha$ is an upper bound on the fraction of malicious clients, and preserves the stated privacy guarantee. For example, tolerating 10\% malicious clients requires only 11\% more fake queries per honest client. See \autoref{sec:security} for details.

\section{Baseline Private Protocol with High Communication}
\label{sec:overview}

We will present our protocol in stages. First, in this section, we present a protocol which provides the required privacy but incurs high communication overhead. This protocol explains sampling of fake queries, and the scheduling of fake and real queries, which help us prove client privacy. Then in Section~\ref{sec:protocol}, we build on this protocol and reduce communication by using SHE. Lastly, Section~\ref{sec:security} contains the proof of security for our final protocol.

We begin by discussing the main ingredients of approximate nearest neighbors search - embeddings and clustering.

\textbf{Embeddings.} Similar to state-of-the-art vector databases and prior PNN protocols, \name\ uses embeddings to represent unstructured data as $d$-dimensional vectors. These vectors preserve semantic similarity, allowing nearest-neighbor search over text, images, or videos. \name\ is compatible with modern embedding models such as BERT~\cite{DBLP:conf/naacl/DevlinCLT19}.

\textbf{Clustering-based nearest neighbors search.} The server stores $N$ entries, each containing a document (e.g., text or image) and associated metadata. In the offline phase, the server maps each document to its embedding. During a query, the client embeds its input and seeks the nearest server embedding and corresponding metadata. Similar to previous PNN protocols, \name\ employs a clustering-based approximate nearest-neighbor search to retrieve relevant results efficiently.
The protocol has the following flow:
\begin{itemize}
 \item \textbf{Initialization.} At initialization, the server divides the document embeddings into $K$ clusters using $K$-means clustering,
 a heuristic-based technique that partitions data into $K$ disjoint clusters $\{C_i\}_{i = 1}^K$, and outputs these clusters and their centroid embeddings $\{c_i\}_{i = 1}^K$. The server also divides the document metadata into the same clusters.
 The server sends the centroid embeddings to the client.
 \item \textbf{Query.} To make a query, the client first locally finds $\Delta$ clusters nearest to the query embedding, $C'_1,\cdots,C'_{\Delta}$.
 For this, the client locally computes the cosine similarity between the query and each of the centroid embeddings and picks clusters
 whose centroid has the highest similarity. The client then sends the nearest clusters $C'_1,\cdots, C'_{\Delta}$ to the server.
 \item \textbf{Response.} The server responds by sending all the embeddings and their metadata in clusters $ C'_1,\cdots, C'_{\Delta}$ to the client.
 \item \textbf{Local client computation.} The client locally computes cosine similarity between the query embedding and each embedding in clusters $ C'_1,\cdots, C'_{\Delta}$, sorts the scores to find the closest embedding, and retrieves the corresponding metadata.
\end{itemize}

Clustering trades accuracy for performance: more clusters reduce computation and response size, but increase the chance of missing the true nearest neighbor. Hence, choosing an appropriate number of clusters is crucial for system usability.

This protocol is not private: to protect the client's query, the protocol must hide the nearest clusters from the server.
Revealing the closest clusters could expose information about the query since the cluster centroids reveal semantic information.

\subsection{Hiding Near Clusters using Differential Privacy}

Hiding the nearest clusters from the server is a key performance bottleneck in previous SHE-based schemes, such as Tiptoe.
Tiptoe does not support querying only the clusters the client is interested in. This is because skipping a cluster will indicate to the server that the client's query is not in those clusters. Thus, to hide the query clusters, the client must send encrypted material for \emph{each} cluster. Similarly, the server performs expensive SHE operations against all the clusters.  Computing over the entire database hides the nearest cluster completely, but this results in prohibitively poor performance.

To overcome this issue, we propose an approach based on differential privacy. 
The key advantage is that \emph{the protocol no longer requires touching all the clusters}, which results in a significant performance improvement.
While differential privacy offers a different privacy guarantee than the full obliviousness provided by Tiptoe, it remains a widely accepted privacy notion. Our approach provides $(\epsilon, \delta)$-DP. Intuitively, this guarantees that the maximum privacy
loss due to the client's queries is essentially bounded by $e^\epsilon$ with probability at least $1-\delta$.

We observe that in large-scale applications, many clients make queries concurrently. In other words, these applications naturally fit a batch processing model, where the server processes a batch of client queries at a time. For such applications, we design a protocol that works on a batch of queries from non-coordinating clients. Our protocol
hides a particular client query within a batch of queries from all participating clients.

Our protocol considers $U$ users and a batch of size $U\cdot\Delta$. We require the following steps.

\textbf{Anonymized queries.} All queries in a batch are anonymized. To achieve this, each client routes their queries through an existing anonymization service. We assume, as is standard, that the anonymization service does not collude with the server. Since many such services already exist and are widely trusted, clients can choose their preferred provider. This makes our protocol more practical than standard multiparty computation, which typically requires dedicated, trusted infrastructure.

Even with anonymized queries, the server may still infer information by analyzing traffic patterns. For instance, if the server knows that only a specific client is interested in a particular cluster, observing access to that cluster could reveal the client's identity. To address this, our protocol ensures that queries are deniable.

To achieve such deniability, we make the following changes:

\textbf{Fake queries.} Each client makes a few fake queries.
 Specifically, for each cluster, the client samples $Z\gets \mathcal{D}$ fake queries.
 We choose a distribution $\mathcal{D}$ such that the server's overall view after seeing all queries is differentially private.
 In expectation, the client makes $\mathbb{E}[\mathcal{D}]$ fake queries per cluster, for a total of $K\cdot\mathbb{E}[\mathcal{D}]$.

\textbf{Epochs.} The protocol operates in fixed-duration \emph{epochs}, each defining a batch of queries from $U$ clients. The protocol guarantees that the server processes all queries by the end of the epoch. Each client receives results by the end of the epoch. Note that even in non-private systems, fixed infrastructure imposes variable query latency; some queries inevitably complete later than others, so epoch-based processing aligns naturally with the resource constraints of large-scale deployments. The per-client overhead of fake queries scales inversely with $U$, so high query volume directly reduces overhead. Epochs do not require synchronization across clients; however, participating clients must stay online throughout an epoch. Each client sends cluster IDs as independent queries rather than together.

\textbf{Timing channel protection.} Even with epochs, the server can exploit arrival times to link queries to specific clients. For example, if a client is known to have a poor network, its queries will tend to arrive towards the end. To suppress this leakage, each epoch is divided into one-second slots. For each query (real or fake), the client independently picks a random slot. This ensures that the server cannot link query timing to a specific client, as the random scheduling and fixed epoch duration decouple query arrival times from client identity.

In \autoref{sec:security}, we prove that sampling the number of fake queries from a negative binomial distribution with parameters $p = e^{-0.2\epsilon/\Delta}$ and $r = 3(1 + \log(1/\delta))/U$, along with a random query schedule, guarantees $(2\epsilon, 2\Delta\delta)$-differential privacy.

\textbf{Drawback.} Even though the protocol conceals the closest clusters, it incurs high communication costs. Concretely, the server transmits embeddings and metadata for all the nearest clusters to the client, resulting in substantial communication overhead.

\textbf{Metadata Fetch.}
After determining the index of the nearest entry, the client retrieves the metadata associated with it. When the metadata per entry is small, the server returns the metadata for the entire cluster with each query. Note that this does not leak any additional information since the server already knows the query cluster. The client then selects the metadata associated with the most relevant entry locally. However, if the metadata is large (in kilobytes), it becomes infeasible to download the entire cluster.

\section{Optimizing Communication using SHE}\label{sec:SHE-instance}
\begin{figure}[!ht]
\centering
  \includegraphics[keepaspectratio,width=\columnwidth]{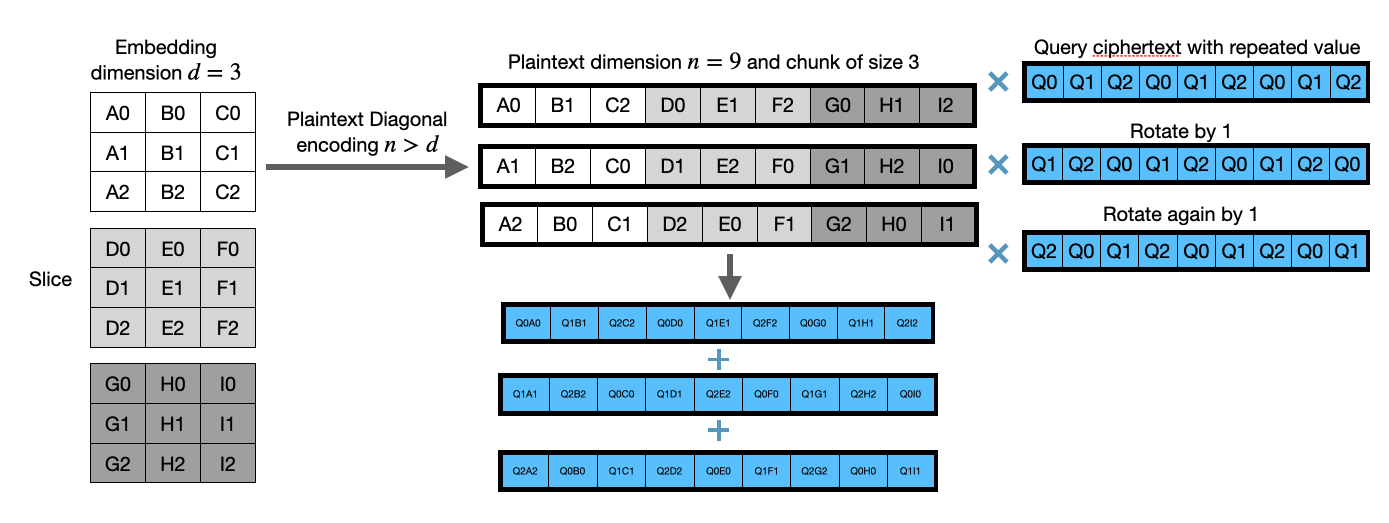}
  \caption{An example dot-product computation for a single cluster with three slices ($d=3$, $n=9$). Plaintext vector
  $j$ packs the $j$-th wrap-around diagonal of every slice. The server aligns the query only once
  for all slices; the output is a single ciphertext whose slots are the dot products of the query with each embedding in
  the cluster.
  }
  \label{fig:pec}
\end{figure}

We use SHE to reduce the communication of this protocol by computing encrypted similarity scores server-side.
In this section, we discuss the details of our efficient SHE-based instantiation as well as describe
our novel optimizations to MulPIR~\cite{DBLP:conf/uss/AliLP0SSY21}, used in metadata fetch for large
metadata entries.

\subsection{Reducing Communication in Nearest Neighbor Search}

At a high level, instead of sending only the closest centroids, each client query (real or fake) includes the query embedding.
An immediate problem with this approach is that these embeddings preserve the semantic meaning of the original data while also containing sensitive information about it. Recent work has shown that embeddings can reveal up to 92$\%$ of the original data~\cite{DBLP:conf/emnlp/MorrisKSR23}.
Therefore, the client cannot send the embedding to the server in plaintext. Instead, the client first encrypts this embedding using SHE.
To maintain anonymity, for each real query, the client generates a new SHE encryption of the query embedding. Similarly, for each fake query, the client generates a new encryption using a fake value.

We rely on vectorized, RLWE-based SHE described in \autoref{sec:background} to hide a query embedding.
At a high level, the client sends an SHE encryption of its embedding to the server, and the server homomorphically
computes similarity scores between the query embedding and all the embeddings in the nearest cluster. The server
then returns encrypted scores to the client.

Looking ahead, the server will need to compute many encrypted cosine similarity scores via SHE dot-products.
We apply the \emph{diagonal encoding} of Halevi and Shoup~\cite{DBLP:conf/crypto/HaleviS18}
in order to efficiently compute these dot-products under SHE: we organize each cluster into
square matrices (slices), then encode these matrices as plaintext vectors according to
their diagonals (\autoref{fig:pec}).

\textbf{Server database.} Recall that the server database is divided into $K$ clusters.
Each cluster consists of at most $N'$ embeddings.
For simplicity, we assume that for a given cluster $C_i$, the dataset is represented as a cube $D_{C_i}$ of
dimensions $\lceil N'/d \rceil \times d \times d$, where $d$ is the embedding dimension.
Each cube consists of $\lceil N'/d \rceil$ slices, each of which is a matrix of dimensions $d \times d$, and each column within a slice is an embedding.
Denote elements of $j$-th embedding as $e_j= [e_j^{0},\cdots, e_j^{d-1}] $.
For now, we assume that the plaintext dimension $n=d$ and each diagonal within a slice is a separate plaintext.
That is, each slice contains $d$ plaintext vectors, and the $j$-th vector consists of entries
\[p_j=[e_0^{j \bmod d}, e_1^{(j+1) \bmod d}, \cdots, e_{d-1}^{(j+d-1) \bmod d}].\]
Equivalently, $p_j$ walks diagonally through the slice: it picks one entry from each embedding, stepping the row index by 1 (with wrap-around).
The client encrypts the query embedding $\hat{q}=[q^0,q^1,\cdots,q^{d-1}]$ in a single SHE ciphertext.

\textbf{Secure dot-product computation.}
Recall from \autoref{sec:background} that the server needs to compute a dot-product to calculate the cosine similarity between two normalized embeddings.
To achieve that, the server must align the query so that each element-wise multiplication pair corresponds to the elements of the query and plaintext vectors, i.e.,
every multiplication is of the form $e_j^{i \bmod d}\cdot q^i$.
The server uses homomorphic rotation to align the elements of the encrypted query vector.
Specifically, for a given slice, starting with the first plaintext vector $j=0$,
the server performs \emph{vectorized ciphertext-plaintext} multiplication between the query ciphertext and the first plaintext vector.
This results in encrypted $p_0 \cdot\hat{q}= [e_0^{0}\cdot q^0, e_1^{1}\cdot q^1, e_2^{2}\cdot q^2, \cdots, e_{d-1}^{d-1}\cdot q^{d-1}]$.
The first multiplication requires no rotation because the query is already aligned with the first vector.
The elements of the second plaintext vector $j=1$ are $p_1 = [e_0^{1}, e_1^{2}, e_2^{3}, \cdots, e_{d-1}^{0}]$.
Hence, the server cannot multiply it directly by the query vector.
The server first homomorphically rotates the query vector one slot to the left, resulting in query vector $\hat{q}^{1}= [q^1,q^2,\cdots,q^{0}]$, then
multiplies the result: $p_1\cdot \hat{q}^{1} = [e_0^{1}\cdot q^1, e_1^{2}\cdot q^2, e_2^{3}\cdot q^3, \cdots, e_{d-1}^{0}\cdot q^{0}]$.
The server then repeats this for all the remaining plaintext vectors in a slice: rotating the query vector in increments
of one and multiplying it by the plaintext vector.
By the end, the server has $p_0\cdot \hat{q}, p_1\cdot \hat{q}^{1}, \cdots, p_{d-1}\cdot \hat{q}^{d-1}$ encrypted products.
The server then sums all of these vectors together, resulting in a single encrypted vector in which each slot is a dot
product between the client query vector and one of the embeddings in the slice.

\textbf{Database and query packing.}
The above description assumes that the plaintext dimension $n$ equals the embedding dimension $d$.
However, $n$ is often significantly larger than $d$. In this setting, each plaintext vector could hold multiple diagonals.
We assume that $n$ is divisible by $d$. When that is not the case, we pad the embeddings to the next integer that divides $n$.
We can pack $n/d$ diagonals across the slices in a single plaintext vector.
From here on, we denote a set of $d$ plaintexts as a chunk, and after encoding, each cluster has $N'/n$ chunks of plaintexts.
Concretely, the $j$-th plaintext vector of the first chunk holds the $j$-th diagonal of the first $n/d$ slices.

The client query encrypts a plaintext vector of $n/d$ repetitions of the query embedding.
The server still aligns the query embedding by homomorphically rotating it left by increments of 1.
This step remains the same as in the previous case.
The server then performs ciphertext-plaintext multiplications and sums plaintexts chunk-wise (summing the first $d$ plaintexts, and so on).
The number of ciphertext-plaintext multiplications is now $\lceil N'/n\rceil \cdot  d$ while the response consists of $\lceil N'/n\rceil$ ciphertexts. In \autoref{fig:pec}, we show an example of packing and multiplication for $n=9$ and $d=3$.

Recall from \autoref{sec:background} that homomorphic rotation is
costly. In the above explanation, the server performs $d$ rotations even after packing to align the query.
In \name\, we use the standard baby-step giant-step (BSGS) optimization~\cite{DBLP:conf/crypto/HaleviS18}
to minimize the number of rotations: for $d$-dimensional inner products, BSGS
requires only $\sqrt d$ rotations for step sizes of $1$ (baby step) and $\sqrt d$ (giant step).
See Algorithm 1 in~\cite{DBLP:conf/crypto/HaleviS18}.

\subsection{Reducing Communication in Metadata Fetch}\label{sec:opt}

When metadata is large (in kilobytes), instead of downloading metadata for the entire cluster, the client obliviously retrieves metadata for only the most relevant entry within a single cluster. To achieve that, we employ keyword PIR.

Our construction is based on cuckoo hashing and SHE-based PIR. Concretely, after receiving encrypted scores, the client finds the most relevant embedding and its cluster locally. The client then uses the identity of this entry to generate a keyword PIR query against that cluster.
The metadata query is then the index of the cluster containing this entry, as well as the PIR keyword query.

\textbf{Protecting privacy in metadata fetch.} Note that we must still protect the query to the cluster.
Therefore, to fetch metadata, clients participate in an additional epoch with $\Delta=1$.
As in previous epochs, we also require anonymized queries; clients generate fake queries and execute them at random slots within an epoch.

\textbf{Details of keyword PIR.}
 The server partitions the metadata entries into $K$ clusters.
 The server encodes each entry as a keyword-value pair, where the keyword is the entry ID and the value is the metadata.
 For ease of explanation, we assume that each cluster consists of at most $N'$ entries.
The server uses cuckoo hashing to map the $N'$ entries
 into a table of size $(1+u)N'$ using two random hash functions $h_1$ and $h_2$, where typically $u \approx 0.5$.
 The guarantee is that if a \emph{keyword} is present in the database,
 it must be located at index $h_1(\textit{keyword})$ or $h_2(\textit{keyword})$ in the table.
 The client uses index PIR to privately fetch entries at these two indices and locally find
 the index containing the entry and associated data.
  \name\ uses MulPIR~\cite{DBLP:conf/uss/AliLP0SSY21} as the underlying index PIR, an efficient single-server PIR scheme based on SHE.
 We emphasize that even though we instantiate \name\ with MulPIR,
 it is compatible with any modern efficient PIR scheme.

  \textbf{Overview of MulPIR.}
 The server arranges the database entries in a $d_1 \times d_2$ matrix,
 where $d_1 \cdot d_2 = N'$.
 For this explanation, we assume that each metadata entry has the same size as the SHE plaintext.
 If entries are larger than a plaintext, then we
 split each entry across multiple plaintexts. If entries are smaller, then we pack multiple entries
 into a single plaintext when applicable.
 A client query in MulPIR is a single
 ciphertext $\hat{q}$ encoding the target entry's indices along both dimensions. Given a query, the server does the following.
  \begin{enumerate}
   \item Expand $\hat{q}$ into two indicator vectors of lengths $d_1$ and $d_2$. The server uses the oblivious expansion
 algorithm given in~\cite{DBLP:conf/uss/AliLP0SSY21}, whose cost is dominated by
   $(d_1 + d_2)$ substitution operations (Galois automorphisms followed by key-switching).
   \item Process the first dimension using $d_1 \cdot d_2$ plaintext-ciphertext multiplications and additions, producing a $d_2$-length ciphertext vector.
   \item Process the second dimension using $d_2$
 ciphertext-ciphertext multiplications and additions, producing a single ciphertext that is the response.
  \end{enumerate}

  \textbf{PIR Optimizations.}
 We present several optimizations to MulPIR that yield a $2\text{-}3\times$ improvement in server computation and communication.

\paragraph{Splitting cuckoo hash tables}
As described above, keyword PIR requires the client to fetch entries at both $h_1(\textit{keyword})$ and $h_2(\textit{keyword})$ in the cuckoo hash table. The cuckoo hash table has $(1+u)N'$ slots for $N'$ entries, so each lookup requires a separate index PIR query on the full table, costing $\approx 2(1+u)N'$ plaintext-ciphertext multiplications total. We optimize this by splitting the database entries into two separate tables, one for $h_1$ and one for $h_2$. Each table contains $\approx N'/2$ entries, and we run a single index PIR query on each, skipping empty slots. This reduces the total cost to $\approx N'$ plaintext-ciphertext multiplications, a $3\times$ improvement for $u = 0.5$.

\paragraph{Uneven dimensions}
For a database with $N' = d_1 \cdot d_2$ entries, the server's computation in MulPIR consists of three types of operations:
\begin{itemize}
  \item Step~1 (expansion): $(d_1 + d_2)$ substitution operations to expand the query ciphertext into indicator vectors.
  \item Step~2 (first dimension): $d_1 \cdot d_2$ plaintext-ciphertext multiplications and additions.
  \item Step~3 (second dimension): $d_2$ ciphertext-ciphertext multiplications and additions.
\end{itemize}
As shown in \autoref{tab:SHE}, plaintext-ciphertext multiplications are inexpensive ($0.02$ ms), while substitution operations cost $t_s = 0.5$ ms and ciphertext-ciphertext multiplications cost $t_m = 2.5$ ms, so $\gamma = t_m / t_s = 5$. The database size fixes the plaintext-ciphertext operations in step~2 and cannot be reduced. However, we observe a trade-off between the substitution operations in step~1 and the ciphertext-ciphertext multiplications in step~3: by choosing uneven dimensions with $d_1 > d_2$, we increase the cheaper substitution operations while reducing the expensive ciphertext-ciphertext multiplications.

Concretely, the server's cost for expansion (step~1) and the second dimension (step~3) is $(d_1 + d_2) t_s + d_2 t_m$.
Dividing by $t_s$ and substituting $\gamma = t_m/t_s$, the total cost in substitution operations is
\[d_1 + (\gamma + 1)d_2.\]
Since $d_1 \cdot d_2 = N'$, we substitute $d_1 = N'/d_2$ to express the cost as a function of $d_2$ alone:
\[\frac{N'}{d_2} + (\gamma + 1)d_2.\]
Minimizing over $d_2$ gives
$d_2 = \sqrt{\frac{1}{\gamma + 1}}\sqrt{N'} \approx 0.41\sqrt{N'}$ for $\gamma = 5$.

This optimization improves MulPIR run times by $30\text{-}40$\%.

\paragraph{Linearizing parts of query expansion}
MulPIR uses an optimized query expansion technique originated by Angel et al.~\cite{DBLP:conf/sp/AngelCLS18}.
In short, the scheme encodes a query index in a plaintext polynomial, then expands it
into $k$ ciphertexts using $k$ ciphertext rotation operations
(substitutions followed by key-switching). We noticed that
some of these expansions can be substituted with linear operations:
call $\ct \pm \ct'$ where $\ct'$ is $\ct$ rotated. In general,
we observe a $10\text{-}25$\% improvement in MulPIR's expansion from our optimization.

\paragraph{Reducing evaluation key size}
To expand a MulPIR query, the server performs repeated substitution operations on the client's query ciphertext using Galois automorphisms $\sigma_k: x \mapsto x^k$ in the ring $R_Q = \Zset_Q[X]/(X^n+1)$. Each substitution requires its own evaluation key sent by the client. Typically, the client sends these keys once, and the server reuses them across queries. However, in \name, reusing keys would allow the server to link queries to the same client, violating our privacy requirement. Therefore, the client must send fresh evaluation keys with each query, and these keys dominate the request size. Note that composing two substitutions yields $\sigma_k \circ \sigma_j = \sigma_{k \cdot j \bmod 2n}$. Therefore, a single key for $\sigma_k$ suffices to compute $\sigma_{k^2 \bmod 2n}$ on any ciphertext by applying the substitution twice.

The expansion algorithm requires substitution keys for Galois elements $g_a = n/2^a + 1$ for
$a = 0, 1, \ldots, c{-}1$, where $c = \lceil \log_2 d \rceil$ and $d$ is the expansion dimension.
For $n = 4096$, these are $g_0 = 4097, g_1 = 2049, g_2 = 1025, g_3 = 513, g_4 = 257, g_5 = 129$.
Note that $g_a^2 \bmod 2n = g_{a-1}$ when $a \leq \lfloor(\ell{-}1)/2\rfloor$ (where $n = 2^\ell$).
For $n = 4096$, this holds for $a \leq 5$, e.g., $129^2 = 16641 \equiv 257 \pmod{8192}$, but fails for $a > 5$.

We therefore use a hybrid approach in which the client sends a single key for $g_5$,
and the server derives $g_0,\ldots,g_4$ from $g_5$ by repeated squaring.
For the remaining higher levels ($g_6, g_7, \ldots$), the client sends explicit keys with each request.
For example, an expansion of dimension $2^{12}$ normally requires $12$ keys; with our optimization, the client sends $7$ instead, reducing communication by approximately $40\%$.

This incurs a small computational overhead. To derive $g_{5-k}$ from $g_5$, the server applies the substitution $2^k$ times instead of once. However, the derived keys are used at early levels of the expansion tree where few ciphertexts exist (levels $0$ through $3$ have $1, 2, 4, 8$ ciphertexts respectively), so the total overhead is small. Concretely, for expanding to $2^{12}$ ciphertexts, the optimization adds $129$ extra substitutions over the normal $4095$, an overhead of approximately $3\%$.

\begin{table}[htbp]
    \centering
    \caption{Experimental computation cost and noise growth
 of each BFV homomorphic operation, measured on an Intel Xeon w3-2423 (single threaded).}
    \begin{tabular}{c  c  c}
    \hline
 Operation & Time (ms) & Noise added (bits) \\
    \hline
    $\mathsf{CtCtAdd}$  & $0.004$ & $0.5$\\
    $\mathsf{PtCtMult}$ & $0.02$ & $20$\\
    $\mathsf{CtCtMult}$ & $2.5$ & $26$\\
    $\mathsf{CtSub}$\footnotemark & $0.5$ & $0.5$\\
    \hline
    \end{tabular}
    \label{tab:SHE}
\end{table}
\footnotetext{Same operation as $\mathsf{CtRotate}$; we use $\mathsf{CtSub}$ for substitution operations in PIR query expansion.}

\begin{figure}[t]

\hrule height 0.4pt

\noindent\textbf{Input:} Semantic search database $\vec{D} = \{(\vec{e}_i, \textsf{meta}_i)\}_{i \in [N]}$, number of clusters $K$

\noindent\textbf{Output:} Plaintext database $\vec{D}_{\text{ptxt}}$, cluster centroids $\vec{c}$

\hrule height 0.4pt

\noindent\textbf{Notation:} Let $g = \sqrt{d}$, $h = \sqrt{d}$. Let $\vec{E}[j][k][l]$ denote the embedding cube, where:
\begin{itemize}
    \item $j \in \left[\left\lceil \frac{|\vec{E}|}{d} \right\rceil\right]$ indexes the slices,
    \item $k \in [d]$ indexes the embedding within a slice,
    \item $l \in [d]$ indexes the slot within an embedding.
\end{itemize}
\begin{algorithmic}[1]
\State $(\vec{c}, \vec{C}) \gets \mathsf{KMeansClustering}(\vec{D}, K)$
\For{$i \in [K]$} \Comment{Iterate over each cluster}
    \State $E \gets \vec{C}[i]$ \Comment{Embeddings cube}
    \State $s \gets \left\lceil |E| / d \right\rceil$ \Comment{Number of slices}
    \State $\vec{P} \gets []$
    \For{$j \in [s]$} \Comment{Iterate over slices}
        \If{$j \bmod s = 0$}
            \For{$k \in [d]$}
                \State $\vec{D}_{\text{ptxt}}[i][\lfloor j/s \rfloor] \gets \mathsf{EncodePtVec}(\vec{P}[k])$
            \EndFor
            \State $\vec{P} \gets []$
        \EndIf
        \For{$l \in [d]$} \Comment{Iterate over embeddings}
            \For{$k \in [d]$}
                \State $\vec{P}[j \cdot d + l][k] \gets \vec{E}[j][k][(k + l) \bmod d]$ \Comment{The $l^{\text{th}}$ diagonal}
            \EndFor
        \EndFor
    \EndFor
    \For{$k \in [s]$}
        \For{$a \in [h]$} \Comment{BSGS giant-step index}
            \For{$b \in [g]$} \Comment{BSGS baby-step index}
                \State $\vec{D}_{\text{ptxt}}[a][k][b] \gets \mathsf{PtRotate}(\vec{D}_{\text{ptxt}}[a][k][b], -g \cdot a)$
            \EndFor
        \EndFor
    \EndFor
\EndFor
\State \Return $\vec{D}_{\text{ptxt}}, \vec{c}$
\end{algorithmic}
\rule{\columnwidth}{0.4pt}
\caption{Server Initialization Algorithm}
\label{alg:servinit}
\end{figure}

\section{\name\ Protocol}
\label{sec:protocol}
In this section, we describe the complete \name\ protocol for a single epoch.
We assume that $U$ honest clients are present to query within that epoch.

\textbf{Server initialization.}
Overview of server initialization is shown in \autoref{fig:serverinit}.
The server database consists of $ N$ precomputed embeddings and their metadata.
The server uses the algorithm in \autoref{alg:servinit}
to encode the database with $N$ embeddings into an index.
Specifically, the algorithm divides the input embeddings into $K$ clusters using K-means clustering.
The algorithm then iterates over each cluster separately.
The $i^{th}$ cluster $C[i]$ consists of $\lceil |C[i]|/d \rceil$ slices.
The algorithm then packs diagonals from $n/d$ slices into $d$ plaintexts.
Once plaintext chunks are generated, the algorithm iterates over each chunk in a group of $g=\sqrt{d}$ plaintexts,
rotating each group in increasing multiples of $-g$. These rotations are done for the baby-step giant-step
(BSGS) optimization~\cite{DBLP:conf/crypto/HaleviS18} that reduces rotations at query time to $\sqrt{d}$.
Note that these rotations are performed on plaintexts, so their cost is negligible.
We only perform these rotations when the number of chunks is less than $\sqrt{d}$.

\textbf{Query generation.} Each client uses the algorithm in \autoref{alg:clientquery} to generate queries
and their schedule at the start of an epoch. The high level overview is shown in the \autoref{fig:dp}
The client locally has a set of cluster centroids $c$ that the server generates during initialization.
The algorithm first generates independent queries for at most $\Delta$ real query embeddings.
To generate a $j$'th real query, the algorithm picks a centroid $\id_l$ nearest to $j$'th query embedding.
Then, the algorithm copies the embedding $n/d$ times into a plaintext vector $\vec{v}$ of size $n$.
This is to take advantage of SHE packing.
The algorithm then generates a fresh SHE secret key $\sk$ and evaluation keys $\evk$,
and encrypts $\vec{v}$ using $\sk$.

Next, the algorithm generates fake queries. Here, the algorithm iterates over each cluster,
sampling a number of fake queries from the negative binomial distribution $\mathcal{NB}(r/U, p)$.
To generate a single fake query, the algorithm generates a fresh
SHE secret key $\sk$ and evaluation keys $\evk$, and encrypts $\vec 0$ using $\sk$.
Note that the client must generate a fresh evaluation key to keep each query anonymous.
Also, each evaluation key must include keys for rotation steps $1$ and $g$,
as server computation rotates by these steps.

The algorithm then permutes the real and fake queries list $\vec Q$ and generates a random schedule for the queries.
To generate the schedule $\vec S$, for each query $i\in Q$, the algorithm picks a random slot $t$
independently and appends $i$ to $S[t]$.
In other words, the schedule $\vec S$ maps slot indices to the query indices.
At a particular time slot $t$ within an epoch, the client will make queries
(independently) indicated by $S[t]$.

\textbf{Server computation.} The server uses the algorithm in \autoref{alg:response}
to generate a response for every query received.
The algorithm first aligns the query ciphertext $\ct$.
That is, first it copies the query ciphertext $g$ times and then rotates each copy left in increments of one, in total $g$ rotations.
Then, the algorithm performs a dot-product between rotated query ciphertexts and the plaintext vectors for cluster $\id$.
Recall that we ensure each cluster is encoded into $d$ plaintexts.
The algorithm iterates in groups of $g$ plaintexts; within each group,
multiply the $i$-th plaintext with the $i$-th
copy of the rotated query and sum the resulting ciphertexts.
These iterations yield $h=\sqrt{d}$ ciphertexts.
After that, the algorithm iterates over the $h$ resulting ciphertexts,
rotating each right in increments of $g$ and summing all of them into a single ciphertext $R$.
This step requires a total of $h$ rotations. The algorithm returns $R$ and the cluster $\id$'s metadata.

\begin{figure}[t]
\hrule height 0.4pt

\noindent\textbf{Input:} Client query $q$, cluster centroids $c = \{\vec{c}_l\}_{l \in [K]}$, number of clients $U$

\noindent\textbf{Output:} Schedule $\vec{S}$
\hrule height 0.4pt
\begin{algorithmic}[1]
\State $\vec{v} \gets []$
\For{$l \in [\lceil n/d \rceil]$} \Comment{Repeat $q$}
    \State $\vec{v} \gets \vec{v} \,||\, \vec{q}$
\EndFor
\For{$l \in [\Delta]$} \Comment{Generate real queries}
    \State $(\sk_l, \evk_l) \gets \kg(1^\lambda)$
    \State $\ct_l \gets \mathsf{EncSemPt}(\sk_l, \vec{v})$
    \State Pick $l$-th nearest cluster $\id_l$ of $\vec{q}$ from $c$
    \State $Q \gets Q \,||\, (\ct_l, \evk_l, \id_l)$
\EndFor
\For{$i \in [K]$}
    \State $F_i \gets \mathcal{NB}(r/U, p)$ \Comment{Number of fake queries}
    \For{$j \in [F_i]$} \Comment{Generate fake queries}
        \State $(\sk'_j, \evk'_j) \gets \kg(1^\lambda)$
        \State $\ct'_j \gets \mathsf{EncSemPt}(\sk'_j, \vec{0})$
        \State Pick random cluster $\id_j$ from $c$
        \State $Q \gets Q \,||\, (\ct'_j, \evk'_j, \id_j)$
    \EndFor
\EndFor
\State $\vec{S} \gets \mathsf{RandScheduleGen}(Q)$ \Comment{Generate schedule}
\State \Return $\vec{S}$
\end{algorithmic}
\rule{\columnwidth}{0.4pt}
\caption{Client Query Generation Algorithm}
\label{alg:clientquery}
\end{figure}

\begin{figure}[t]
\hrule height 0.4pt

\noindent\textbf{Input:} List $\vec{Q}$ of real and fake queries, epoch length $T$

\noindent\textbf{Output:} Schedule $\vec{S}$

\hrule height 0.4pt

\begin{algorithmic}[1]
\State Initialize a schedule $\vec{S}$ as $T$ empty lists
\For{$j \in [|\vec{Q}|]$}
    \State $i \gets \text{UniformRandom}([T])$ \Comment{Sample random slot}
    \State $\vec{S}[i] \gets \vec{S}[i] \,||\, j$
\EndFor
\State \Return $\vec{S}$
\end{algorithmic}
\rule{\columnwidth}{0.4pt}
\caption{Random Schedule Generation Algorithm}
\label{alg:schedule}
\end{figure}


\begin{figure}[t]

\hrule height 0.4pt

\noindent\textbf{Input:} Query $(\ct, \evk, \id)$

\noindent\textbf{Output:} Encrypted scores and data $r$

\hrule height 0.4pt

\begin{algorithmic}[1]
\For{$j \in [g]$} \Comment{Align query}
    \State $\ct'[j] \gets \mathsf{CtRotate}(\ct, j, \evk)$
\EndFor
\State $\vec{D} \gets \vec{D}_{\text{ptxt}}[\id]$
\State $\vec{R} \gets []$
\For{$k \in [m]$} \Comment{Iterate over chunks}
    \State $S \gets 0$ \Comment{Initialize accumulator}
    \For{$i \in [h]$} \Comment{Giant step}
        \For{$j \in [g]$} \Comment{Baby step summation}
            \State $S \gets S + \mathsf{CtRotate}(\vec{D}[k] \cdot \ct'[j], g \cdot i, \evk)$
        \EndFor
    \EndFor
    \State $\vec{R} \gets \vec{R} \,||\, S$
\EndFor
\State $\textsf{data} \gets \textsf{data}_{\id}$
\State \Return $r := (\vec{R}, \textsf{data})$
\end{algorithmic}
\rule{\columnwidth}{0.4pt}

\caption{Server Computation Algorithm}
\label{alg:response}
\end{figure}
\begin{figure}[htbp]
    \centering
    \includegraphics[keepaspectratio,width=\textwidth]{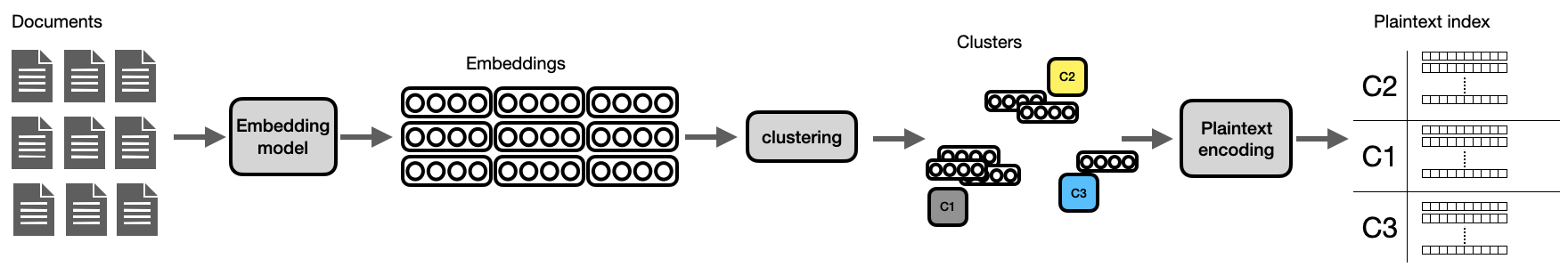}
     \caption{Overview of server initialization. First, the server converts documents into embeddings, then clusters these embeddings, and then encodes them into an index consisting of SHE plaintexts.}
    \label{fig:serverinit}
  \end{figure}

  \begin{figure}[htbp]
    \centering
    \includegraphics[keepaspectratio,width=\textwidth]{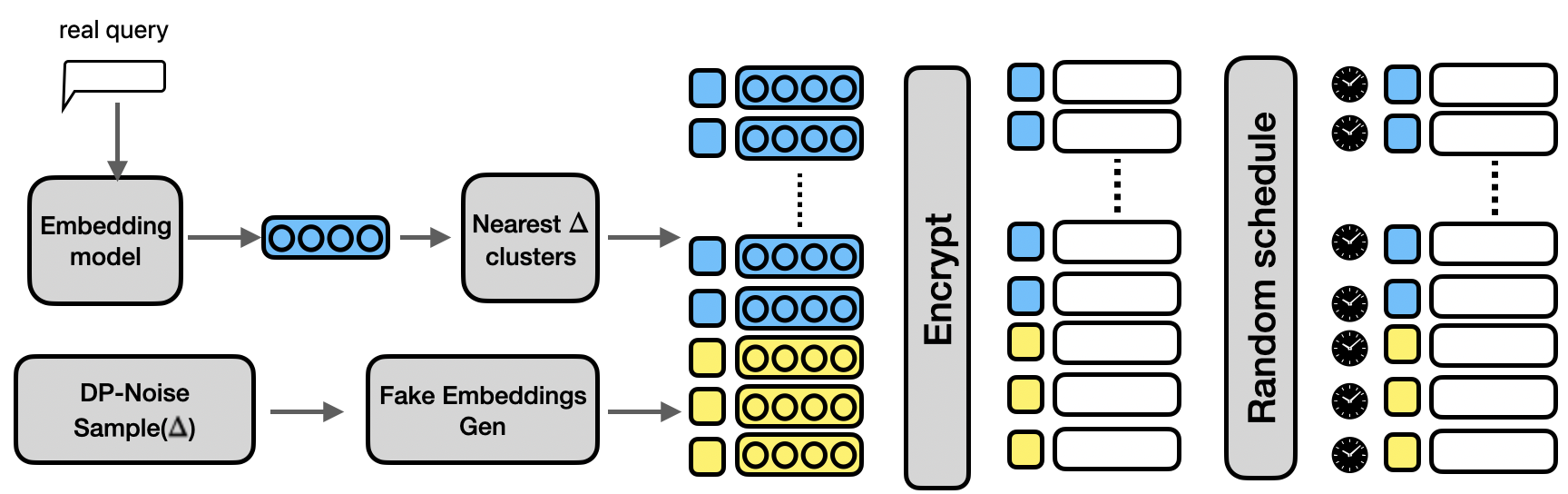}
    \caption{   Overview of query generation at client for a single epoch. The client first converts query into embedding. Then it finds nearest $\Delta$ clusters and creates $\Delta$ real queries. It then adds fake queries, encrypts them using SHE, and generates a random schedule for them. The client then executes each query independently at random instances within an epoch.}
    \label{fig:dp}
  \end{figure}
\section{Security}\label{sec:security}
In this section, we prove that if clients follow
\autoref{alg:clientquery} to generate queries, then the overall system achieves
$(2\epsilon, 2\Delta\delta)$-differential privacy in each epoch. Note from \autoref{alg:clientquery} that each query
consists of the nearest centroid and the embedding. We only focus on proving the
privacy of centroids because the IND-CPA security of SHE protects query embeddings.
Therefore in this section, we model each query as a nearest centroid without an embedding.
During each epoch, the server receives multiple queries for each cluster. Towards the end of the epoch,
the server's view can be considered as a noisy
histogram over all clusters. We consider this as an output of the DP mechanism.

We prove this argument in two steps:
First, we show that a noisy histogram output by a curator,
which gets users' queries as input and uses negative binomial $\mathcal{NB}$ mechanism to sample fake queries,
is $(2\epsilon, 2\Delta\delta)$-DP in the central model.
Second, we show the view of the server in an epoch can be simulated from the curator's output.

\textbf{Privacy due to malicious clients.} We assume all the clients participating in a batch will follow the protocol. However, a fraction $\alpha$ of malicious clients could negatively affect privacy by not submitting fake queries. In the distributed setting, each honest client samples fake queries from $\mathcal{NB}(r/U, p)$ per cluster, and the sum over $U$ honest clients yields $\mathcal{NB}(r, p)$. If $\alpha U$ clients are malicious and contribute no fakes, the total fake noise per cluster reduces to $\mathcal{NB}((1-\alpha)r, p)$, degrading the privacy guarantee. To compensate, honest clients can over-provision by sampling from $\mathcal{NB}(r/(U(1-\alpha)), p)$ instead of $\mathcal{NB}(r/U, p)$. The sum from $(1-\alpha)U$ honest clients then recovers $\mathcal{NB}(r, p)$, preserving the original $(2\epsilon, 2\Delta\delta)$-DP guarantee. The cost is a $1/(1-\alpha)$ factor increase in fake queries per honest client. For example, tolerating $\alpha = 0.1$ (10\% malicious clients) requires only 11\% more fake queries per honest client.

\subsection{DP Security in the Central Model}
We define the central model as follows:

\textbf{Central model.}
The curator performs the following:
\begin{enumerate}
 \item Collects the real queries from all the clients, $\{R_u\}_{u\in [U]}$, where
 $U$ is the total number of clients, and $R_u$ is each client's set of
 queries and will be at most $\Delta$. Let $R$ denote the list of all real queries.
 \item For each cluster $b$, samples the number of fake queries $F_j \gets \mathcal{NB}(r,p)$.
 Add it to the fake queries list $F$.
 \item Given the list of real and fake queries, randomly permute them:
 $\mathsf{list} \gets \mathsf{permute}(R, F)$.
 \item Send $\mathsf{list}$ to the server. \label{line:list}
\end{enumerate}
Another way to view it is that the curator sends a noisy histogram of queries to each cluster to the server.
Recall that each client can contribute at most $\Delta$ queries.

Differential privacy  in the central model is stated in the following theorem.
\begin{theorem}\label{thm:DP}
 For any query database of size $N$, number of clusters $K$, differential privacy parameters
 $\epsilon,\delta \in (0,1)$, and $\Delta \in \mathbb{N}$, let
 $p = e^{-0.2\epsilon/\Delta}$ and $r = 3(1 + \log(1/\delta))$. Then,
 the negative binomial mechanism $\mathcal{NB}(r,p)$ is $(2\epsilon, 2\Delta\delta)$-DP
 for $\Delta$-histogram.
\end{theorem}

\begin{proof}
    We define the curator as Algorithm $M$ that takes input from a
    database $X$ of requests from $N$ clients and outputs a noisy histogram.
    We prove the theorem by investigating the affected buckets
    when we change the input $X$ by replacing one client with another.

    For a particular cluster $b$ in the histogram, we define
    $M_b(\cdot) := R_b + \mathcal{NB}(r,p)$, where $R_b$ is
    the total real queries for the cluster $b$ in the database.
    Note that this is the curator's exact output for each cluster.
    We break down the proof into two edge cases: 1) the simpler
    case where $X$ and $X'$ differ by a client which sends all
    its messages to a single cluster and 2) they differ by
    a client which sends $\Delta$ messages to $\Delta$ different
    clusters.

    \textit{Simple case.}
    Consider two neighboring databases
    $X$ and $X'$ which differ in queries contributed by a single client. Concretely,
    in $X$ the client contributes $\Delta$ queries to a cluster $b$
    and in $X'$ it contributes to another cluster $b'\not = b$.
    When $j \notin \{b,b'\}$ then $M_j(R)$ and $M_j(R')$ are identically distributed.
    Thus, the privacy provided by $M_j$ is 0-DP for all $j \notin \{b,b'\}$.

    Now consider the cluster $b$; the two databases differ by $\Delta$ real queries for this cluster.
    By the following, lemma $M_b$ provides $(\epsilon, \delta)$ privacy on cluster $b$.

    \begin{lemma}[Theorem 13~\cite{DBLP:conf/icml/Ghazi0MP20}]\label{thm:centralDP}
    For any $\epsilon,\delta \in (0,1)$ and $\Delta \in \mathbb{N}$, let
    $p = e^{-0.2\epsilon/\Delta}$ and $r = 3(1 + \log(1/\delta))$. Then,
    the $\mathcal{NB}(r,p)$-Mechanism is $(\epsilon, \delta)$-DP
    in the central model for $\Delta$-summation\footnote{$\Delta$-summation is the counting problem in which each client can send at most $\Delta$ queries.
    We extend it to histograms, $\Delta$-histograms, or the counting problem adopted by allowing each client to make at max $\Delta$ queries per epoch across all buckets.
    }.
    \end{lemma}

    We can make a similar argument for a cluster $b'$. Therefore, by basic composition,
    Lemma~\ref{lemma:basic_comp}, the total privacy provided by $M$ for these neighboring databases is $(2\epsilon, 2\delta)$-DP.

    \textit{General case.}
    Now, we consider a general case where two neighboring databases that differ in queries from a single client.
    In database $X$ the client contributes queries $\gamma_j$ for cluster $j$ and in database $X'$
    the client contributes queries $\gamma_j'$ for the cluster $j$.
    Define $ \Delta_j := |\gamma_j - \gamma_j'|$.
    The total difference between two neighboring
    databases is given by $\sum_{j}^{B}\Delta_j\leq 2 \Delta$.

    Let  $\epsilon_j := \frac{\epsilon \Delta_j}{\Delta} \leq \epsilon$.
    Further, we only consider the clusters where $\Delta_j > 0$ since the
    clusters where $\Delta_j = 0$ are already $0$-DP. (Equivalently, $\epsilon_j= 0$ for these $j$.)
    Then, Lemma~\ref{thm:centralDP} implies cluster $j$ is $(\epsilon_j, \delta)$-DP
    for $\Delta_j$-summation for the same mechanism since
    \[\frac{\epsilon}{\Delta} = \frac{\epsilon}{\Delta}\cdot\frac{\Delta_j}{\Delta_j} = \frac{\epsilon_j}{\Delta_j}.\]
    Simple composition over all the clusters yields
    \[\sum_j \epsilon_j = \frac{\epsilon}{\Delta} \sum_j \Delta_j \leq 2\epsilon.\]
    Basic composition also yields that the composed protocol has
    $\delta' := 2\Delta \delta$ since at most $2\Delta$ clusters
    differ.
\end{proof}

We emphasize that the theorem is stated for a particular epoch $S$. If clients participate in
many epochs, say $l$, then basic composition yields $(2l\epsilon, 2l\cdot\Delta\delta)$-DP.
Tighter accounting via advanced composition~\cite{DBLP:conf/stoc/DworkRV10} or R\'{e}nyi DP~\cite{DBLP:conf/csfw/Mironov17} would improve the $\epsilon$ dependence from $l$ to $O(\sqrt{l})$, which is a direct improvement for clients who participate in many epochs. We leave this analysis to future work.

\subsection{Generating View of the Server}
We show how the curator's output is sufficient to simulate the server's view generated by distributed DP
algorithm defined in \autoref{alg:clientquery} together with an anonymization network (AN).
We also require pseudorandom ciphertexts and evaluation keys,
a common assumption for RLWE-based SHE schemes. This is
also why we cannot reuse evaluation keys.
This statement is stated in the following theorem.

\begin{theorem}\label{thm:sim}
    Let $\mathcal{AN}$ be an anonymous network that 1) removes all identifying information
    from messages it receives and 2) randomly permutes all the messages it receives
    over a second-long slot. Let each honest client run the distributed mechanism in \autoref{alg:clientquery}
    with messages sent to the server via $\mathcal{AN}$. Then, the view of the server in an epoch can be simulated
    using the curator noisy histogram output described above in the central model.
\end{theorem}

\begin{proof}
We use infinite divisibility of $\mathcal{NB}$ distribution
to produce a standard cryptography hybrid argument.

{\bf Hybrid 0:} Every client is the same as the real client given in Algorithm~\ref{alg:clientquery} except it does not include
 encrypted embeddings with each request (real or fake).
 Hybrid 0 is equivalent to the real client finding the nearest centroid for each real query and
 a fake random centroid for each fake query. The client then picks a random slot for each query and send
 associated queries at each slot.
\\

{\bf Hybrid 1:} In this hybrid, we replace all the clients with a single simulator. The input of a simulator is a list
$R = \{R_u\}_u$, where $R_u$ are the real queries for $u$-th client. The simulator then simulates each client $u$
with input $R_u$ as in Hybrid 0.

As the simulator internally runs each client $u$ as in
Hybrid 0 with input $R_u$, from the adversary's point of view,
the distribution of requests is the same as Hybrid 0.
\\

{\bf Hybrid 2:} In this Hybrid, the simulator's input is the same as in the previous Hybrid 1.
The simulator is defined as follows.
\begin{enumerate}
    \item  Collects the real queries from all the clients, $\{R_u\}_{u\in [U]}$.
    Call it a list of real queries $R$.
    \item For each cluster $b$, samples the number of fake queries $F_j \gets \mathcal{NB}(r, p)$.
    Add it to the fake queries list $F$. \label{line:large_nb}
    \item Given the list of real queries $R$ and fake queries $F$, randomly permutes them:
        $\mathsf{list} \gets \mathsf{permute}(R, F)$ \label{line:outcur}.
    \item Initialize schedule list $S$ of size $z$ (where z is the total number of slots in an epoch).
    For each query $q \in \mathsf{list}$, pick a random slot $p \gets z$ and set $S[p]= S[p]\cup q$.  \label{line:schedule}
    \item For each slot $p \in S$ send queries in $S[p]$ to the server.
\end{enumerate}
There are two differences between Hybrid 2 and Hybrid 1:
\begin{enumerate}
    \item In Hybrid 1, each simulated client samples fake queries for each cluster from a distribution $\mathcal{NB}(r/U,p)$. While
    in Hybrid 2 (Line~\ref{line:large_nb}) the simulator samples fake queries per cluster from $\mathcal{NB}(r,p)$.
    \item In Hybrid 1, each simulated client generates a schedule for its queries independently, while in Hybrid 2, the simulator
    generates a schedule for queries.
\end{enumerate}
To show that Hybrid 2 and 1 are indistinguishable, we use the definition of infinitely divisible
distribution, Definition~\ref{def:inf_div}.

Note that negative binomial distribution is infinite divisible because $\mathcal{NB}(r,p) = \sum_{i=1}^{U} \mathcal{NB}(r/U,p)$.
Thus, the distribution of per cluster fake queries per cluster in Hybrid 2  (Line~\ref{line:large_nb}) is the same as in Hybrid 1.
Observe that a random slot is picked for each
query independently in both Hybrids.
Therefore, the query schedule generated in both Hybrids is indistinguishable.

To complete the proof, observe that in Line~\ref{line:outcur} of Hybrid 2, the simulator generates $\mathsf{list}$ in a manner similar to the curator in the central model. The key difference is the scheduling: the simulator assigns each query to a uniformly random slot (Line~\ref{line:schedule}), while the curator sends $\mathsf{list}$ all at once.

We argue that the server's per-slot view can be simulated from the global histogram. Let $H = (H_1, \ldots, H_K)$ denote the global histogram of cluster counts in $\mathsf{list}$. In Line~\ref{line:schedule}, each query is assigned to one of $z$ slots independently and uniformly at random. Therefore, for each cluster $b$, the $H_b$ queries to cluster $b$ are distributed across slots according to a multinomial distribution $\mathsf{Multinomial}(H_b; 1/z, \ldots, 1/z)$. This distribution depends only on $H_b$ and $z$, both of which are determined by the global histogram and public parameters. Within each slot, the AN randomly permutes all queries. Since the per-slot assignment is a public-randomness function of the global histogram, and the within-slot ordering is uniformly random, the server's complete per-slot view (counts and ordering) can be simulated given only the global histogram $H$ and public randomness. No additional per-client signal remains.

\end{proof}

\section{Implementation}

We discuss implementation details for \name's components.

\subsection{Open Source SHE Library}
\label{sec:open-source-she-library}

Components of \name\ outlined in \autoref{sec:SHE-instance} are instantiated using the BFV encryption scheme~\cite{DBLP:journals/iacr/FanV12}.
We have released the SHE components of \name\ as \swifthe~\cite{\swifthecitation}, an open-source library
available at \href{\swifthelink}{\swifthelink}. \swifthe\ implements the BFV scheme, a secure dot-product for similarity compute, and keyword PIR for metadata fetch. The library API is simple and flexible, and can be used for other applications.  We also provide an easy-to-follow tutorial and examples to use the API.

We discuss key optimizations and implementation decisions:
\begin{itemize}
    \item \textbf{Residue Number System (RNS) variant of BFV}. The vanilla BFV scheme requires manipulating elements modulo the ciphertext modulus $Q$, which is an integer, possibly with hundreds of bits. Implementing this scheme requires slow multi-precision modulo arithmetic. To address this, we instead implemented the RNS variant of BFV~\cite{bajard2016full}.
    We also implemented the hybrid RNS variant of key-switching~\cite{kim2021revisiting} and fast base conversion~\cite{DBLP:conf/sacrypt/BajardEHZ16} methods required for ciphertext-ciphertext multiplication and rotation operations.

    \item \textbf{Reducing the computation of BFV operations}. The number theoretic transformation (NTT) speeds up polynomial multiplications required in most BFV operations. We implemented the NTT using Harvey’s butterfly~\cite{harvey2014faster} and optimized it using lazy modular reduction. This optimization improves forward and inverse NTT runtime by $20\text{-}25$\%.
    For the many sequential ciphertext-plaintext multiplications required in similarity computation and MulPIR, we apply the same lazy reduction strategy to avoid expensive modular operations.
    In the second dimension of MulPIR, we perform many ciphertext-ciphertext multiplications. We compute the tensor product and addition in a large basis, then scale back to modulus $Q$. When performing BFV multiplications in sequence, we skip the scaling step after each multiplication and perform it only once at the end~\cite{DBLP:conf/asiacrypt/KimPZ21}. This results in $10\text{-}20$\% improvement in MulPIR~\cite{DBLP:conf/uss/AliLP0SSY21} runtime.

     \item \textbf{Reducing the response size}. To reduce the response ciphertext size, the server uses the modulus switching technique
     to mod switch the ciphertext down to a single RNS limb. Furthermore,
     as the plaintext is encoded in higher order bits, the server only serializes the bits needed for decryption. Following the calculation in~\cite{DBLP:conf/uss/HuangLHD22}, we set $z=8$ and drop $l_0= \log_2{(q_l/t)}-2$ from the first ciphertext polynomial and $l_1\leq \lfloor\log_2(q_l/t)-2-\log_2(z \sqrt{2n/9})\rfloor$ from the second polynomial. This parameter selection gives a decryption error of $2^{-37.5}$.


    \item \textbf{Reducing the request size}.
    As discussed in \autoref{sec:SHE-instance}, evaluation keys dominate the request size. To minimize key size in similarity compute, the client transmits only two keys: one for the giant step and another for the baby step (the server repeatedly uses the same keys for all rotations).

\end{itemize}

\subsection{Embedding and Clustering}

We map each document in the server database into dense embeddings that retain semantic information. For this, we use the sentence embedding model \emph{msmarco-distilbert-base-tas-b} with 66 million parameters, optimized for semantic search~\cite{Hofstaetter2021_tasb_dense_retrieval}. In particular, we use the implementation provided by the `sentence-transformers' Python library\cite{reimers-2019-sentence-bert}. The model outputs dense 768-dimensional floating point vectors. We then use PCA to reduce the dimensionality of each vector to 192.

The next step is clustering these vectors. To do so, we first find $K$ cluster centroids by running k-means clustering. Then we assign each embedding to the cluster with the closest centroid. Finally, all points and their assigned cluster are compiled into an index.

We implemented this pipeline using the open-source library FAISS~\cite{douze2025faisslibrary}. FAISS provides fast algorithms for k-means clustering and indexing. FAISS exposes a flexible clustering API that allows us to experiment with various clustering parameters, such as the number of data points sampled, the number of iterations of the basic k-means algorithm. For indexing, we utilize one of the Flat indexes provided by FAISS to perform fast intra-cluster nearest neighbor search.

\subsection{Parameters Selection}
\textbf{SHE parameters.}
Performance of \name\ directly depends on our choice of BFV parameters.
We evaluated multiple parameter sets and selected the one that yielded optimal performance.
For search and metadata fetch, we use: polynomial degree $n=4096$, ciphertext modulus of $\log_2{Q}= 83$ bits, with CRT limbs of $27, 28$~ and~$28$ bits. We sample error from a centered binomial distribution with standard deviation $3.2$ and secret key from a ternary distribution. This achieves 128 bits of post-quantum security~\cite{DBLP:conf/uss/AlkimDPS16}.

For search, we set plaintext modulus of $\log_2{t}= 16$ bits.
Recall that the embedding model outputs floating-point vectors of dimension $d$; however BFV operates over integer vectors. We convert each float into a fixed-point integer by multiplying it by the appropriate scaling factor $p$ and then rounding to the nearest integer. We select $p$ to ensure the dot products do not wrap around modulo $t$. We say that precision is equal to $\log_2(p)$ bits. To avoid wrap-around, it is important that $p < \sqrt{(t-1)/2} - \sqrt{d}/2 $, which implies that $\log_2(p) < 1/2 \log_2 t  \approx 7$. Therefore, we obtain at least 7 bits of precision. To increase the precision, we use the plaintext CRT decomposition.
For this, we pick two plaintext CRT limbs of plaintext modulus of $\log_2{t_0}= 16 \text{ and } \log_2{t_1}= 17$ bits. This yields a total of 31 plaintext bits, i.e., 15 bits of precision, without increasing evaluation key size, since the BFV evaluation key doesn't depend on the plaintext modulus $t$.

As mentioned in \autoref{sec:open-source-she-library}, we reduce the size of the response ciphertext by first mod switching down to the last CRT limb of 27 bits and then further dropping the least significant bits. To avoid decryption error due to dropping bits, we set $l_0 = 9$ and $l_1=0$. This drops the size of the single response ciphertext from 55 KB to 23.5 KB.
For MulPIR to handle larger noise growth, we use a plaintext modulus of $\log_2{t}= 5$ bits. Having a smaller plaintext modulus allows us to set $l_0 = 19$ and $l_1=13 $, reducing each response ciphertext to 12 KB.

\textbf{Number of users and epochs.}\label{ssec:epoch}
Our protocol batches queries from $U$ users. A larger $U$ reduces the number of fake queries required per user, which in turn improves the performance. We set $U$ to 100K and 500K in our experiments. Another parameter is epoch length. Recall that each participating client sends queries at random time slots within an epoch, so some clients may wait until the end of the epoch. Therefore, a long epoch length will negatively affect client experience. We assume the server provisions 10,000 cores, consistent with large-scale search deployments. The epoch length is determined by the total number of queries (real and fake) divided by the server's aggregate throughput. For our parameter settings, this ranges from tens of seconds to a few minutes.

\textbf{Differential privacy parameters.}
For a fixed failure probability $\delta$, a smaller $\epsilon$ means stronger privacy but worse performance since each client must now execute more fake queries. We set the mechanism parameters to $\epsilon_0=0.05$ and $\delta_0=2^{-30}$. By \autoref{thm:DP}, the per-epoch guarantee is $(2\epsilon_0, 2\Delta\delta_0) = (0.1,\; 2\Delta\cdot 2^{-30})$-DP. For $\Delta \leq 5$ this gives $\delta \leq 2^{-26}$; for $\Delta \leq 10$, $\delta \leq 2^{-25}$.

For each cluster, a client samples fake queries from $\mathcal{NB}(r/U,p)$. This means in expectation the number of fake queries is

\begin{equation*}
\mathbb{E}[\text{Fake queries per client}] = \tfrac{rpK}{(1-p)U} = \tfrac{3(1+\log_2(1/2^{-30}))e^{-0.0005/\Delta}K}{(1-e^{-0.0005/\Delta})U}
\end{equation*}

Recall $K$ is the total number of clusters, $U$
is the number of participating honest clients, and $\Delta$ is the number of clusters each client probes. A straightforward calculation shows that the expected number of fake queries per client is at most $0.006 \Delta K$ for $U = 500{,}000$, a modest estimate for large-scale applications~\cite{semrush2025}. For example, with $K = 256$ clusters and $\Delta = 1$, each client generates fewer than $2$ fake queries on average. \autoref{tab:bench_wally} reports exact values for varying configurations.


 \textbf{Search parameters and metric.}

The product of clusters $K$ and the number of clusters probed $\Delta$ determine the performance and accuracy of \name. A smaller $K$ and a smaller $\Delta$ both enhance performance by reducing the overhead of fake queries. Further notice that very small $K$ will result in large response size.  Similarly, reducing $\Delta$ may lead to a substantial drop in accuracy due to the increased probability of missing the true nearest neighbor.

Our objective is to minimize the fake queries overhead, which is proportional to the product $K\Delta$, while balancing search quality and response size. We experimented with different parameters and selected those that provide the best accuracy while keeping $K$ moderately small and minimizing $\Delta$ as much as possible. For each probe, the client generates a new query with a large evaluation key. Consequently, a large $\Delta$ will result in large total request size.

\textbf{Dataset.} We use the MS MARCO document ranking dataset ~\cite{bajaj2016ms}. The dataset contains approximately 3.2 million passages, where each passage comprises a document identifier, document URL, document title, and document body. The embedding is generated from the document title and body.

\begin{table*}[!htbp]
\small
    \caption{Comparison of \name\, Tiptoe and Pacmann for MS MARCO~\cite{bajaj2016ms} database. Communication and QPS for \name\ includes overhead due to fake queries. Worst value in each metric is marked in red. For \name\ calculation we assume $U=500,000$. Epoch size ranges from 20 seconds to two minutes for varying values of $\Delta$.}
       \centering
    \makebox[\textwidth][c]{
\resizebox{0.8\textwidth}{!}{
    \begin{tabular}{l  c  c  c c }
        \hline
     & Client Storage (MB) & Communication (MB) & Queries Per Second (QPS) & MRR@100 \\
    \hline
    Tiptoe~\cite{DBLP:conf/sosp/HenzingerDCZ23}  &  0.61 & 17.4 & \cellcolor{red} 909 & \cellcolor{red} 0.11 \\
    Pacmann~\cite{DBLP:journals/iacr/ZhouSF24} & \cellcolor{red} 614 & \cellcolor{red} 61.6 & 34,482 & 0.26  \\
    \name\ ($K=256$, $\Delta = 1$) & 0.04 & 0.56 & 25,974 & 0.12 \\
    \name\ ($K=256$, $\Delta = 3$)& 0.04 & 1.7 & 9,881 & 0.16 \\
    \name\ ($K=256$, $\Delta = 5$)& 0.04 & 2.6 & 6,667 &  0.18\\
    \hline
    \end{tabular}
}}
    \label{tab:main_exp}
\end{table*}

\section{Evaluation}\label{sec:performance}
We demonstrate \name's concrete performance and compare it with baselines.
We ran server-side computations on a 6-core Intel Xeon w3-2423 with 32GB of RAM.

We measure performance via the following metrics. All metrics account for the overhead of fake queries: we amortize the expected number of fake queries over each real query. For example, if a client makes 1 real query and 2 fake queries, all three contribute to the reported communication and server computation per query.
\begin{itemize}
    \item Queries per second (QPS): We measure latency for a single query on a single core and compute QPS as 10,000 divided by that latency, assuming a server with 10,000 cores. This extrapolation is valid because each query is processed independently with no inter-query synchronization. Since QPS scales linearly with core count for all schemes, relative comparisons are independent of this choice. Any shared-resource effects such as memory bandwidth contention apply equally across schemes.
    \item Bandwidth: We calculate the size of the client query and the server response for each query.
    \item Accuracy: To evaluate the accuracy of our search results, we used the Mean Reciprocal Rank (MRR@100) metric, a standard measure for assessing the quality of search results. Given a query $q$ and a ground truth index $i$ such that $DB[i]$ is the most relevant entry, we merge all scores (along with entry indices) obtained from probing $\Delta$ clusters, sort them, and retain the first 100 scores. If the list contains $i$ at the $j$-th rank (where $j \leq 100$), the reciprocal rank score is calculated as $1/j$. In the absence of $i$ in the list, MRR@100 is set to 0.
\end{itemize}

\textbf{Baselines.} We compared \name\ with the following baselines. Both baselines offer full obliviousness but have high performance overhead:

{\bf Tiptoe (simulated)}~\cite{DBLP:conf/sosp/HenzingerDCZ23}, which is a state-of-the-art private search system based on linear homomorphic encryption from the learning with errors problem~\cite{DBLP:journals/jacm/Regev09}. Similar to \name, Tiptoe uses a clustering approach to find the nearest candidate. However for each query it only probes a single cluster. Running the complete Tiptoe system requires substantial infrastructure. We simulate Tiptoe using favorable estimates that give it an advantage in our comparison.
    \begin{itemize}
        \item Latency: We estimate the minimum time the server will take to process a single query. We exclude network cost and the per-query offline phase, both of which further favor Tiptoe.  For each query, the computation is dominated by dot-products between vectors.
        To privately fetch and rank candidates from a particular cluster, the server performs $N$ dot-products between vectors of dimension $d$ (ranking). We assume that each component of a vector is 32 bits.
        \item Accuracy: We divide the database into $\sqrt{N}$ clusters, probe the nearest cluster and calculate MRR@100 of these candidates.
        \item Communication: We calculate the minimum number of 32-bit elements the client and server have to transfer per query. To generate a token, the client sends at least $n=2048$ RLWE ciphertexts or $n^2$ elements. For ranking, the client sends a query consisting of $d\sqrt{N}$ elements. For tokens, the server response consists of $16$ RLWE ciphertexts (each at least 2048 elements) and for ranking the response consists of $\sqrt{N}$ elements.

    \end{itemize}

    \begin{table*}[!htbp]
    \caption{Performance of \name\ across varying database sizes. Each client will probe a single cluster $\Delta=1$. Each setting provides $(\epsilon=0.1, \delta = 2^{-26})$-DP guarantee.}
    \small
    \centering
    \begin{tabular}{l  c  c  c c c c }
        \hline
     \textbf{Entries in database N} &  \multicolumn{2}{c}{\textbf{1 Million}}   & \multicolumn{2}{c}{\textbf{16 Million}} & \multicolumn{2}{c}{\textbf{100 Million}}  \\
    \hline
    Number of users $U$  & 100,000 & 500,000 & 100,000 & 500,000& 100,000 & 500,000\\
    Number of clusters ($K$)  & 128 & 256 & 256 & 512 &  512 & 1,024 \\
    Expected fake queries     &  1.7 & 0.3 & 3.3 & 1.7 &  6.7 & 2.6 \\
    Request size (MB)  &  0.76 & 0.36 & 1.18 & 0.76 & 2.17 & 1.01 \\
    Response size (MB) &  0.17 & 0.04 & 2.3  & 1.5 & 13.9 & 3.03  \\
    Queries per second (in thousand) &  22 & 66 & 1.4 & 3.2 & 0.47 & 1.39 \\
    Epoch size   &  10 sec  & 10 sec  &  2 min  & 3 min  &  4 min  & 6 min\\
    \hline
    \end{tabular}
    \label{tab:bench_wally}
\end{table*}

\begin{table}[!htbp]
\small
    \caption{Performance of \name\ for different values of $K$ with $U=100{,}000$, $\Delta=3$, and one million database entries. Epoch sizes range from 10 to 20 seconds. Each setting provides $(\epsilon=0.1, \delta = 2^{-26})$-DP guarantee.}
    \centering
    \begin{tabular}{l   c c c c  }
        \hline
    Number of clusters ($K$) & 64 & 128 & 256 & 512 \\
        \hline
    Request size (MB)  & 0.45 & 0.45 & 0.9 & 1.35 \\
    Response size (MB) & 0.10 & 0.10 & 0.2 & 0.3  \\
     QPS (in thousand) & 31 & 31 & 22 & 15 \\
    \hline
    \end{tabular}
    \label{tab:bench_wally_K}
\end{table}

\begin{table}[!htbp]
\small
    \caption{Performance of \name\ for different values of $\Delta$ with $U=100{,}000$, $K=128$, and one million database entries. Epoch sizes range from a few seconds to about a minute. Each setting provides $(\epsilon=0.1, \delta \leq 2^{-25})$-DP guarantee.}
    \centering
    \begin{tabular}{l   c c c c  }
        \hline
    Clusters probed ($\Delta$)  & 1 & 3 & 5 & 10 \\
        \hline
    Request size (MB)  &  0.36 & 0.45 & 0.59 & 0.65 \\
    Response size (MB) & 0.08 & 0.10 & 0.11 & 0.14  \\
     QPS (in thousand) & 37 & 31 & 27 & 21 \\
    \hline
    \end{tabular}
    \label{tab:bench_wally_delta}
\end{table}

\textbf{Pacmann}~\cite{DBLP:journals/iacr/ZhouSF24}, which is concurrent work. Like \name, this scheme employs sublinear server computation, but it offers full obliviousness. Pacmann's approach differs from \name\ in two key aspects. First, to locate the nearest neighbor, Pacmann utilizes a graph-based search algorithm. Second, Pacmann employs sublinear Offline/Online PIR to retrieve graph nodes and traverse the graph. For each query, the scheme requires multiple round trips and per-query maintenance. We compare our results with the ones provided in the Pacmann paper~\cite{DBLP:journals/iacr/ZhouSF24}, which is sufficient since we conducted our own experiments in a comparable environment.

\textbf{Metadata Fetch.} To fetch metadata, Tiptoe relies on an offline/online PIR scheme called SimplePIR (referred to as the URL step in their paper). We run the SimplePIR open-source implementation to estimate performance for this step. Pacmann has not implemented metadata fetch step. The comparison for metadata fetch is given in \autoref{sec:metadatafetch}.

\textbf{Comparison with Tiptoe.} The results are presented in \autoref{tab:main_exp}. \name\ significantly outperforms Tiptoe across all metrics. Notably, \name’s QPS is $7\text{-}29\times$ higher. This is because in Tiptoe, the server must perform at least one expensive operation per entry in the database. Similarly, the communication per query is $6.7\text{-}31\times$ smaller than Tiptoe's, for varying numbers of clusters probed. Tiptoe can only probe a single cluster per query, and to probe more clusters, a new query must be initiated. For $\Delta=1$, a single probe case, \name’s MRR@100 score is only slightly better than Tiptoe’s score of 0.11. However, as we increase to $\Delta=5$, \name’s MRR@100 score improves to 0.18, which is considerably better than Tiptoe’s score.

\textbf{Comparison with Pacmann.} \autoref{tab:main_exp} shows that Pacmann has higher QPS and MRR@100 than \name. Specifically, the higher QPS is because Pacmann is based on offline/online sublinear PIR and their higher MRR@100 is due to the use of a graph-based search data structure, which is superior to the clustering technique. However, this scheme requires large client storage and communication per query. Specifically, for a database with 3.2 million entries, the client must store a hint of 614 MB and per-query communication is 61.6 MB, which is prohibitively high, especially for cellular networks. Another downside is that hint-based protocols cannot directly handle database updates. Finally, to gain better MRR@100 the scheme requires multiple round trips, which add to the network delay. We did not consider network delay in Pacmann's QPS calculation. Specifically, due to its use of offline/online PIR, QPS is $1.3\text{-}5.6\times$ higher than \name\ but communication is $23\text{-}123\times$ higher, depending on $\Delta$.

\textbf{Benchmarking \name.}
In \autoref{tab:bench_wally}, we evaluate the performance of \name\ across various database sizes ranging from one million entries to one hundred million entries and cluster counts ranging from 128 to 1,024. For all these experiments, we assume that the client probes the nearest cluster $\Delta=1$. We consider $U = 100{,}000$ and $U = 500{,}000$ honest users, with epoch sizes ranging from 10 seconds to 6 minutes.

First, we notice that increasing the number of honest users $U$ leads to improved overall performance across all database sizes. Specifically, the QPS increases by $2\text{-}3\times$, while the request and response sizes improve by $1.5\text{-}2.4\times$. This demonstrates that \name\ is an ideal choice for applications with moderate to extended delays.

Second, we observe that the expected number of fake queries is nearly proportional to the product of the number of honest users ($U$) and inverse of the number of clusters ($K$). For each database, as $U$ increases from 100,000 to 500,000, we increase the number of users by $5\times$ (which reduces fake queries) but also increase the number of clusters by $2\times$ (which increases fake queries). Consequently, we see that fake queries decrease by a factor of almost $5/2 = 2.5$, indicating a linear dependency. Alternatively, we can consider that reducing the number of clusters $K$ for a fixed database will lead to a decrease in fake queries.

Third, we note that as the database size varies, the request and response sizes linearly increase. However, the QPS decreases. This is because for larger databases, each cluster size becomes substantial, resulting in increased server operations per query.
\autoref{tab:bench_wally_K} shows the effect of varying $K$ with $\Delta=3$ for a one-million-entry database. For small $K$ (64 and 128), the overhead of fake queries is negligible and performance remains constant. As $K$ grows, communication and QPS degrade gradually. For instance, quadrupling $K$ from 128 to 512 increases request size by $3\times$ and reduces QPS by $2\times$. Even at $K=512$, the system still achieves 15,000 QPS with 1.35 MB request size, demonstrating that \name\ scales gracefully with the number of clusters.

\autoref{tab:bench_wally_delta} shows the effect of varying $\Delta$ with $K=128$. Increasing $\Delta$ from 1 to 10 increases request size by only $1.8\times$ and decreases QPS by $1.8\times$, a sub-linear cost relative to the $10\times$ increase in probes. Meanwhile, accuracy improves significantly. From \autoref{tab:main_exp}, MRR@100 increases from 0.12 ($\Delta=1$) to 0.18 ($\Delta=5$), closing the gap with Pacmann's MRR@100 of 0.26, which requires 614 MB of client storage and 61.6 MB of communication per query. This demonstrates a practical throughput-accuracy tradeoff that system designers can tune based on application requirements.


\begin{table}[h]
    \caption{Comparison of \name\ and SimplePIR for fetching metadata of size 1KB from a database with 3.2 million entries divided into $K=256$ clusters. Communication and QPS for \name\ include overhead due to fake queries.}
    \centering
    \begin{tabular}{l    c  c  }
        \hline
     & Communication (MB) & Queries Per Second (QPS) \\
    \hline
    SimplePIR~\cite{DBLP:conf/sosp/HenzingerDCZ23}  &   40.1 &  85,327  \\
    \name\ ($\Delta = 1$) &  1.8 & 35,211 \\
    \hline
    \end{tabular}
    \label{tab:metadata_fetch}
\end{table}
\paragraph{Metadata retrieval performance}
\label{sec:metadatafetch}
Recall that when metadata size is large, clients participate in an additional epoch with $\Delta=1$. That is, the client will only access the cluster in which the most relevant entry falls. They will still have to make additional fake queries and follow a random schedule. Within each query cluster (real or fake), the server performs keyword PIR computation using SHE. In Tiptoe, the server has to perform SimplePIR computation over the entire database to fetch metadata. \autoref{tab:metadata_fetch} shows the comparison of \name\ and SimplePIR for metadata fetch for a database with 3.2 million entries, where each metadata is 1KB. Note that while SimplePIR has a 2.4$\times$ higher QPS, its communication overhead is 22$\times$ higher than \name. We emphasize that for SimplePIR, we excluded the offline hint download cost.
\section{Related Work}\label{sec:relatedwork}
\paragraph{Homomorphic encryption}
HE alone is a popular technique used in private search
systems~\cite{DBLP:conf/sosp/HenzingerDCZ23,DBLP:conf/sosp/AhmadSAAG21,DBLP:journals/tbbis/EngelsmaJB22,DBLP:conf/btas/Boddeti18}.
These systems achieve full obliviousness but require encrypted
computation over the entire database, resulting in high bandwidth and server cost that scales linearly with database size. In contrast, \name\ computes
on only a few clusters per query, achieving sublinear server cost with
DP to hide the requested cluster index from the server.
\autoref{sec:intro} provides an in-depth comparison with Tiptoe~\cite{DBLP:conf/sosp/HenzingerDCZ23}.
A related system is~\cite{DBLP:journals/pvldb/PangDX10} which sends embellished text queries in the clear but uses additively HE~\cite{benaloh1994dense} to compute an encrypted similarity score.

\paragraph{Two-server and MPC-based systems}
Two-server protocols such as Preco~\cite{DBLP:conf/sp/Servan-Schreiber22} achieve full cryptographic obliviousness via distributed point functions and probabilistic batch codes, but require 10~MB of communication and 10 seconds of latency for a database with 10M entries, with linearly-sized computation per server. SANNS~\cite{DBLP:conf/uss/0030CDPRR20} uses garbled circuits and HE but sends GBs of data per query for databases of 1M entries or more. Both systems require two non-colluding servers that jointly perform the computation. \name\ replaces this with a single computation server and a lightweight anonymization network, which is easier to deploy since widely trusted anonymization services already exist. The tradeoff is that \name\ provides differential privacy rather than exact obliviousness; in return, \name\ achieves $7\text{-}29\times$ higher QPS and orders-of-magnitude lower communication than these systems.

\paragraph{Concurrent works}
Pacmann~\cite{DBLP:journals/iacr/ZhouSF24} efficiently reduces graph-based private nearest neighbors search
to many PIR queries on a graph. They use offline-online PIR and each PIR computation is sublinear by way of a sublinear PIR~\cite{DBLP:conf/sp/ZhouPZS24}.
Pacmann is a cryptographic, fully-oblivious private search system and offers good accuracy
via graph-based nearest neighbors search, which results in higher MRR@100 than clustering-based approaches.
However, Pacmann's graph traversal requires many sequential round trips, each involving a PIR query, which adds significant network latency. Additionally, Pacmann requires 614~MB of client storage and 61.6~MB of communication per query, which is prohibitively high for mobile clients and cellular networks. Clients must also stream the entire database offline and the hints must be refreshed on database updates. In contrast, \name\ requires only 0.04~MB of client storage, $0.56\text{-}2.6$~MB of communication, and no offline phase.

Panther~\cite{DBLP:conf/ccs/LiHZHLWC25} uses a clustering-based approach similar to SANNS but replaces expensive DORAM with batch PIR based on SealPIR~\cite{DBLP:conf/sp/AngelCLS18}, and combines secret sharing, garbled circuits, and HE.
Panther is a cryptographic, fully-oblivious single-server private search system that also provides server database privacy, ensuring the client learns only the final top-$k$ results.
However, Panther is an interactive protocol requiring multiple rounds of secret sharing and garbled circuit evaluation, and its latency degrades significantly under WAN conditions (e.g., $9\text{-}18$ seconds for 1M--10M entries). Additionally, Panther requires 93--284~MB of communication per query for 1M--10M entries. In contrast, \name\ is non-interactive and achieves orders-of-magnitude lower communication ($0.56\text{-}2.6$~MB) with thousands of QPS.

\paragraph{Hint-based and recent single-server PIR}
Recent single-server PIR schemes achieve high throughput through preprocessing. SimplePIR~\cite{DBLP:conf/uss/HenzingerHCMV23} and Piano~\cite{DBLP:conf/sp/ZhouPZS24} require clients to download large offline hints (121~MB to 724~MB for 1--32~GB databases in SimplePIR) that must be refreshed when the database changes. YPIR~\cite{DBLP:conf/uss/MenonW24} eliminates the offline hint download but has large query upload sizes (846~KB to 2.5~MB for 1--32~GB databases) that grow with the database size. In \name, metadata PIR operates within a single cluster of size $N' \ll N$, so MulPIR's linear-in-$N'$ computation is already efficient without requiring any offline phase, client-side storage, or large uploads. For these reasons, \name\ uses MulPIR for metadata retrieval, though it is compatible with any single-server PIR scheme.

\paragraph{DP and the shuffle model}
\name's DP construction is closely related to the shuffle model of differential privacy~\cite{DBLP:conf/eurocrypt/CheuSUZZ19,DBLP:journals/corr/abs-2001-03618}, where clients send messages through a shuffler that permutes them before the server observes them. The encode-shuffle-analyze paradigm, exemplified by Prochlo~\cite{DBLP:conf/sosp/BittauEMMRLRKTS17}, follows this pattern for analytics tasks. \name\ instantiates a similar structure for private search: clients send real and fake queries through an anonymization network that acts as the shuffler. However, in \name\ clients locally add noise (fake queries) rather than relying on the shuffler to inject noise. A key difference from typical shuffle-DP systems is that a dedicated shuffler aggregates all messages before forwarding, adding latency equal to the entire collection period. \name\ instead relies on existing lightweight anonymization services that forward messages continuously within one-second batches, allowing the server to process queries as they arrive rather than waiting for the full epoch.

DP information-theoretic PIR
schemes~\cite{DBLP:journals/popets/ToledoDG16,DBLP:conf/uss/AlbabIVG22} achieve similar privacy goals but rely on
multiple non-colluding servers.
PIR in the shuffle model~\cite{DBLP:conf/citc/IshaiKL024,DBLP:conf/ccs/GasconIKL0024}
is closer to \name\ in that it requires one PIR server, replacing DP
with statistical indistinguishability.
However,~\cite{DBLP:conf/citc/IshaiKL024} is a theoretical work,
while~\cite{DBLP:conf/ccs/GasconIKL0024} relies on non-standard cryptographic assumptions
and requires many sub-queries per query, incurring a large response size.

Further related works are discussed in \autoref{sec:relatedwork_appendix}.
\subsection{Further Details}\label{sec:relatedwork_appendix}
\paragraph{Early private search systems}
Early private search systems heuristically add fake queries or alter queries client-side~\cite{DBLP:journals/oir/Domingo-FerrerSC09,DBLP:conf/cse/YeWPC09,DBLP:conf/sdm/MurugesanC09,DBLP:journals/tit/Rebollo-MonederoF10,DBLP:journals/pvldb/PangDX10,DBLP:journals/corr/abs-1109-4677,DBLP:conf/sp/BalsaTD12,DBLP:journals/ir/ArampatzisDE15,DBLP:conf/trustcom/PetitCMBK15}.
They are broken by a semi-honest server, i.e., the server can de-anonymize
queries~\cite{DBLP:journals/corr/abs-1211-0320,DBLP:journals/jcs/PeddintiS14,DBLP:conf/ccs/GervaisSSCL14,DBLP:journals/jisa/PetitCBMCBK16}
since queries are sent in the clear.
or they have strict distribution requirements on queries, e.g., i.i.d.~queries~\cite{DBLP:journals/tit/Rebollo-MonederoF10}.
In contrast, \name\ has a differential privacy guarantee: the server's view between
histograms over clusters with and without a client's queries is differentially private.
Also, client queries in \name\ are encrypted  preventing the
attack in~\cite{DBLP:journals/popets/ToledoDG16}.

\paragraph{Anonymization networks}
Anonymization networks and mix-nets~\cite{DBLP:conf/uss/DingledineMS04,DBLP:journals/cacm/Chaum81} anonymize
the client but send the query in the clear. Only applying ANs to search systems reveals clients' information to a semi-honest server~\cite{DBLP:conf/trustcom/PetitCMBK15}.

\paragraph{Trusted hardware}
Trusted hardware is another tool used in private search systems~\cite{DBLP:conf/middleware/MokhtarBFPPS17,DBLP:conf/icdcs/GoltzschePMBBFK18,DBLP:conf/osdi/LiZZ0XA021}.
However, these search systems are only as strong as their underlying hardware~\cite{DBLP:conf/uss/BulckMWGKPSWYS18,DBLP:journals/ieeesp/MurdockOGBPG20,DBLP:conf/sp/GrussLSGJOSY18,DBLP:conf/uss/ChenVMDOG21,DBLP:journals/ieeesp/ChenCXZLL20}.
\section{Conclusion}
\name\ is the first protocol to demonstrate batch private nearest-neighbor search at large scale. By combining somewhat homomorphic encryption with differential privacy, it protects client queries and releases only a differentially private histogram of aggregate traffic. \name\ achieves over 25K QPS with low communication overhead and accuracy on par with state-of-the-art systems.


\bibliographystyle{ACM-Reference-Format}
\bibliography{style/refs}

\end{document}